\documentclass{fac}

\usepackage{times}
\usepackage{stmaryrd}
\usepackage{amsmath}
\usepackage{color}
\usepackage{url}
\usepackage{oz}
\usepackage{opsem}
\usepackage{drafty}
\usepackage{graphicx}

\usepackage{stmaryrd}
\SetSymbolFont{stmry}{bold}{U}{stmry}{m}{n}

\usepackage{fix-cm}

\usepackage{traces}
\usepackage{rgref}
\usepackage{rcp}

\makeatletter
\newcommand{\Zset}[3]{\{ #1 \ifx \@empty#2\else | #2 \fi \ifx \@empty#3\else \spot #3 \fi \}}
\makeatother

\renewcommand{\spot}{\mathrel{\cdot}}

\newcommand{\eseq}{\seqT{}}
\newcommand{\popempty}{\mathsf{none}}
\newcommand{\ceval}[2]{\Meaning{#1}_{#2}}

\newcommand{\notr}{\overline{r}}

\makeatletter
\newcommand{\seqex}[1]{[\ifx \@empty#1\@empty ~ \else #1 \fi]}
\makeatother
\newcommand{\seqT}[1]{\seqex{#1}}
\newcommand{\Trt}[2]{(#1,#2)}
\newcommand{\Trf}[2]{\Trt{#1}{\seqex{#2}}}
\newcommand{\Tra}[2]{(#1 #2)}
\newcommand{\emptyclosure}[1]{empty\_closure(#1)}
\newcommand{\prefixclosure}[1]{prefix\_closure(#1)}
\newcommand{\abortclosure}[1]{abort\_closure(#1)}

\newcommand{\STest}[1]{[#1]}
\newcommand{\botvalue}{\bot_{v}}
\newcommand{\botstate}{\bot_{s}}

\newcommand{\arelidle}[1]{\atomicrel{#1}_{idle}}
\newcommand{\arelidler}{\arelidle{r}}

\mathchardef\mhyphen="2D

\renewcommand{\Env}{\mathop{\Keyword{\epsilon \mhyphen rely}}}
\newcommand{\envc}[1]{(\Env #1)}
\newcommand{\envcr}{\envc{r}}
\newcommand{\eguard}[1]{(\mathop{\Keyword{\epsilon \mhyphen restrict}} #1)}
\newcommand{\pguard}[1]{(\mathop{\Keyword{\pi \mhyphen restrict}} #1)}
\newcommand{\eguardr}{\eguard{r}}
\newcommand{\pguardg}{\pguard{g}}
\newcommand{\optr}[1]{opt(#1)}

\newcommand{\assertp}{\Pre{p}}

\newcommand{\quint}[5]{\{#1, #2\} #5 \{#3, #4\}}
\newcommand{\quintprgq}[1]{\quint{p}{r}{g}{q}{#1}}
\newcommand{\quintprgqc}{\quintprgq{c}}

\renewcommand{\Post}[1]{\Sspec{}{}{#1}}

\newcommand{\SPost}[1]{\Post{#1}}
\newcommand{\SPostq}{\SPost{q}}

\newcommand{\myparagraph}[1]{\paragraph{#1.}}
\newcommand{\refsect}[1]{Sect.~\ref{s:#1}}

\newcommand{\pl}{\parallel}
\newcommand{\asgn}{\mathop{:\!=}}

\newcommand{\Meaning}[1]{[\hskip -2pt [#1 ]\hskip -2pt ]}

\newcommand{\refeqn}[1]{(\ref{#1})}

\newcommand{\barx}{\overline{x}}

\newcommand{\idbarx}{\id(\barx)}

\newcommand{\choice}{\mathrel{\sqcap}}

\newcommand{\Choice}[3]{\overset{#1}{\underset{#2}{\Nondet}} #3}

\newcommand{\State}{\Sigma}
\newcommand{\Command}{\mathit{Com}}

\newcommand\reffig[1]{Figure~\ref{fig:#1}}

\newcommand\IF[3]{\If #1 \Then #2 \Else #3}

\renewcommand\Var{\mathop{\Keyword{vardecl}}}
\newcommand{\With}{\mathop{\Keyword{with}}}
\newcommand{\Resource}{\mathop{\Keyword{resource}}}

\newcommand{\Local}[2]{\mathop{\Keyword{local}} #1 \spot #2}

\newcommand{\cnext}{\boldsymbol{\atnext}}
\newcommand{\calways}{\boldsymbol{\always}}
\newcommand{\ceventually}{\boldsymbol{\eventually}}

\renewcommand{\Nil}{\cgdd}
\renewcommand{\Magic}{\boldsymbol{\top}}
\renewcommand{\Abort}{\boldsymbol{\bot_{\pi}}}
\newcommand{\EAbort}{\boldsymbol{\bot_{\epsilon}}}
\newcommand{\Fair}{\Keyword{fair}}
\newcommand{\Fairterm}{\Keyword{fairterm}}

\newtheorem{theorem}{Theorem}
\newtheorem{lemma}[theorem]{Lemma}
\renewenvironment{proof}{\par\noindent\textbf{Proof.}}{\noindent$\Box$}

\newcommand{\encode}[1]{[\![#1]\!]}

\newcommand{\strictconjunction}{weak conjunction}
\newcommand{\Strictconjunction}{Weak conjunction}

\edef\today{\number\day\ \ifcase\month\or
  January\or February\or March\or April\or May\or June\or
  July\or August\or September\or October\or November\or December\fi
  \ \number\year}
\newcounter{Hours}
\newcounter{Minutes}
\newcommand{\CurrentTime}{%
 \ifthenelse{\value{Hours}<10}{0}{}\theHours:%
 \ifthenelse{\value{Minutes}<10}{0}{}\theMinutes}

\title[Designing a wide-spectrum semantics \draftonly{(\today)}]
{%
Designing a semantic model for a \\
wide-spectrum language with concurrency
}
\author[R. Colvin, I. J. Hayes and L. Meinicke \draftonly{(\today)}]
{
Robert J. Colvin \and
Ian J. Hayes \and
Larissa A. Meinicke\\
School of Information Technology and Electrical Engineering,\\ 
The University of Queensland, Australia}

\date{\today\ \CurrentTime}

\begin{document}

\maketitle

\begin{abstract}
A wide-spectrum language integrates specification constructs into a programming language
in a manner that treats a specification command just like any other command.
This paper investigates a semantic model for a wide-spectrum language that supports concurrency
and a refinement calculus.
In order to handle specifications with rely and guarantee conditions,
the model includes explicit environment steps as well as program steps.
A novelty of our approach is that we define a set of primitive commands and operators,
from which more complex specification and programming language commands are built.
The primitives have simple algebraic properties which support proof using algebraic reasoning.
The model is general enough to specify notions as diverse as
rely-guarantee reasoning, temporal logic, and progress
properties of programs, and supports 
refining specifications to code.
It also forms an instance of an abstract concurrent program algebra, 
which facilitates reasoning about properties of
the model at a high level of abstraction.
\end{abstract}


\section{Introduction and motivation}
\label{s:intro}

Concurrent programming languages and their semantics are relatively well understood
but the problem of how to abstractly specify concurrent programs remains a challenge.
For a sequential program one way to abstractly specify its behaviour is via 
a precondition-postcondition pair \cite{Floyd67,Hoare69a}.
For concurrent programs, attending only to initial and final states
is inadequate:
one needs to express 
the intermediate behaviour of a program and its interaction with its environment
along with any assumptions it makes about its environment.

\subsection{Goals}

The goal of this research is to provide a semantic model for a
\emph{concurrent wide-spectrum language} that includes both
specification constructs and standard programming combinators and
primitives.  The benefit lies in expressing related concepts from
concurrency in a uniform manner, and thus supporting reuse of
theories.

The model supports a \emph{refinement calculus} \cite{BackWright98}
with a refinement relation, $c_1 \refsto c_2$, representing that
specification $c_1$ is refined (or implemented) by $c_2$.
The refinement calculus approach is compositional in the sense that
refining a subcomponent refines the whole, for example, if $c_1
\refsto c_2$ and $d_1 \refsto d_2$ then the parallel composition of
$c_1$ and $d_1$, $c_1 \pl d_1$, is refined by $c_2 \pl
d_2$.\footnote{In other words, the operators of the calculus are
  \emph{monotonic} in their arguments with respect to the refinement
  ordering.}
Extending the refinement calculus to handle concurrency in a
compositional manner requires a rich specification language that can
state the desired behaviour of a program with respect to assumptions
about the behaviour of its environment.  The underlying model of a
command is based on traces that distinguish program and environment
steps originally suggested by Aczel~\cite{Aczel83} and further
developed by de Roever and
others~\cite{BoerHannemanDeRoever99,DeRoever01}.  The approach uses a
small set of primitive commands and basic operators to define more
complex commands.

The original motivation for this work was to provide a general and
abstract theory to support the \emph{rely-guarantee} approach for
reasoning about concurrent programs of Jones \cite{Jones81d,Jones83b}
but the resulting theory is applicable to other concepts including
atomic operations and
temporal logic specifications \cite{Pnueli77}.
The model is an instance of a concurrent program algebra~\cite{AFfGRGRACP-TRX},
and so it can be used to prove consistency of the axioms given
there. The algebra can then be used to verify properties of the model
at a higher level of abstraction. The algebra of Hayes~\cite{AFfGRGRACP-TRX} is
similar to the Concurrent Kleene Algebra of Hoare et
al.~\cite{DBLP:journals/jlp/HoareMSW11}, although they are
distinguished by their subtly different axiomatizations. In general,
program algebras can be used to unify diverse programming models
through an axiomatization of commonalities, allowing proofs of
fundamental program properties to be reused across program semantics
that instantiate the algebra. The simple algebraic characterisation
also facilities a straight-forward mechanisation of those proofs.

\subsection{Wide-spectrum language}

We motivate our semantic model and choice of primitive constructs
by considering a range of
specification and programming constructs they should be rich enough to express.

\myparagraph{Preconditions} 
A specification $c$ may be weakened by an assumption 
that its initial state satisfies a precondition $p$. 
In the sequential refinement calculus this is handled by a precondition command $\Pre{p}$
that terminates successfully if $p$ holds, otherwise it aborts,
i.e.\ it allows any behaviour whatsoever.
The weakened specification is written ``$\Pre{p} \scomp c$'', 
the sequential composition of $\Pre{p}$ and $c$.
The key point is that an implementation (refinement) $d$ of $\Pre{p} \scomp c$ 
must behave as $c$ when $p$ holds initially but is unconstrained
otherwise.

\myparagraph{Postconditions}
In the sequential refinement calculus a specification $\Post{q}$, in
which $q$ is a binary relation, represents a requirement that if
started in a state $\sigma$ it must terminate in a state $\sigma'$
such that $(\ssp) \in q$. 
A similar construct is required in the
context of concurrency but implementing such a specification is
challenging because of possible environment interference (see \refsect{end-to-end}). 
For this, stronger assumptions about the environment, such as rely conditions, are required
(see below). 
The sequential specification $\Post{q}$
only needs to consider the initial and final states and hence $q$ is a
binary relation. More generally, the execution of a program makes a
sequence of steps and hence a specification could be
generalised to constrain the sequence of steps. One such
generalisation is to handle Jones' guarantee conditions (see below).

\myparagraph{Rely conditions}
To provide a compositional approach to refining concurrent programs,
Jones \cite{Jones81d,Jones83b} introduced a rely condition, 
a binary relation $r$,
that represents an assumption that the interference the environment makes
between any two atomic steps of the program satisfies $r$ between its
before state and after state. 
Because there may be no environment
steps between two program steps, Jones requires $r$ to be reflexive,
i.e.\ for all states $\sigma$, $(\sigma,\sigma) \in r$,
and because there may be multiple environment steps between 
two program steps, Jones requires $r$ to be transitive,
i.e.\ $(\ssp) \in r$ and $(\sigma', \sigma'') \in r$ imply $(\sigma,\sigma'') \in r$ (see Figure~\ref{f:rely-guar}). 
Therefore, relation $r$ is equal to its reflexive, transitive closure $r^{\star}$.

\begin{figure}
\begin{center}
\input{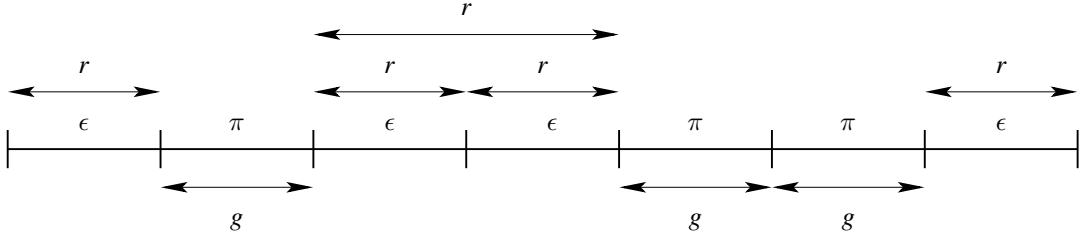}
\caption{Trace of execution with environment steps satisfying $r$ and program steps satisfying $g$}\label{f:rely-guar}
\end{center}
\end{figure}

Handling rely conditions in the semantics requires the ability to refer explicitly to
environment steps as well as program steps. 
The model presented in Section~\ref{s:semantics-new} is rich enough to allow one to define a command $\envc{r}$,
that represents an assumption that all environment steps satisfy the relation $r$.
In the same way that a precondition command $\Pre{p}$ aborts if $p$
does not hold initially, $\envc{r}$ aborts if the environment makes a step 
not satisfying $r$ (see \refsect{rely}).
An environment assumption $\envc{r}$ is added to a specification using the 
\strictconjunction\ operator (see \refsect{strict-conjunction}).

\myparagraph{Guarantee conditions}
If a program $c_2$ with a rely condition $r$ is part of the environment of a program $c_1$,
then $c_1$ should respect the rely condition $r$ of $c_2$,
i.e.\ every atomic program step taken by $c_1$ should satisfy $r$.
To handle this, Jones introduced a guarantee condition
--- a binary relation $g$ ---
and requires every atomic program step made by $c_1$ to satisfy $g$. 
As long as the guarantee $g$ of $c_1$ is contained in the rely $r$ of $c_2$, 
$c_1$ forms a suitable (part of the) environment for $c_2$.

A guarantee condition forms a restriction on the allowable steps made by a program.
The semantic model is rich enough that one can define a guarantee command
$\pguard{g}$, that forces every program step to satisfy the relation $g$ (see \refsect{guarantee}).
Such a guarantee can be combined with a specification using \strictconjunction\ 
(see \refsect{strict-conjunction}).
Guarantee conditions may be generalised to allow more flexible constraints on
the allowable sequences of steps.

\myparagraph{Atomic specifications and linearisability}
A \emph{linearisable} operation \cite{linearisability} must (appear to) take
place atomically between its invocation and response.  It does not alter
shared data outside of that single atomic step, although it may 
modify local data and perform other local computation
(see \refsect{linearisability}).
The framework straightforwardly handles atomic specifications, and distinguishes 
them from multi-step specifications such as postconditions described above.

\myparagraph{Temporal logic}
The semantic model is rich enough to encode temporal logics, such as linear temporal logic \cite{Pnueli77}.
Because the semantics distinguishes program and environment steps,
one can extend the temporal logic to distinguish between properties of the environment
and of the program, as well as relating the two.

\myparagraph{Expression evaluation}
In the context of concurrency, the evaluation of an expression is subject to interference from concurrent
programs modifying variables used within the expression
and hence expression evaluation cannot be assumed to be atomic. 
Expressions are used in assignment statements and control structures
such as conditionals ({\bf if}) and repetitions ({\bf while}). 
In the sequential refinement calculus a conditional control structure can be defined
in terms of a \emph{test} primitive $\STest{b}$, that succeeds if $b$ evaluates to true 
but is infeasible otherwise.
It can be defined as follows, where ``$\nondet$'' represents non-deterministic choice.
\begin{eqnarray*}
  \IF{b}{c_1}{c_2}  & \sdef &  (\STest{b} \scomp c_1) \choice (\STest{\neg b} \scomp c_2)
\end{eqnarray*}
Our approach preserves this form of definition by redefining the meaning of a
test to be non-atomic and hence allow for interference during expression evaluation (see \refsect{expressions}).

\subsection{Contributions}

The main contributions of this paper are 
(i)
a semantics that supports concurrency for both programming language and specification constructs
and hence reasoning about programs in a refinement calculus style, and 
(ii)
a theory of atomic program and environment steps that treats environment steps as first-class citizens.
The theory is based on a small set of primitive commands and
operators, allowing the development of algebraic properties of
operators. The semantics also forms an instance of a concurrent
program algebra~\cite{AFfGRGRACP-TRX}, which facilitates reasoning about
properties of the model at a high level of abstraction.
The theory supports rely and guarantee constructs, and it 
handles progress properties, e.g.\ temporal logic, including
extensions to handle properties of the environment as well as the
program. 

Our semantic model, primitive commands, and program combinators are
defined in Sections~\ref{s:semantics-new}, \ref{s:primitives} and
\ref{s:combinators}, respectively. The wide-spectrum language is
introduced in Section~\ref{s:wide-spectrum}, and we discuss how it can
be used to reason about progress properties in
Section~\ref{s:progress}. 
We conclude in Section~\ref{s:conclusions}.

\section{Semantic model}\label{s:semantics-new}

In our semantic model, program behaviours are described using traces
of primitive steps.  Section~\ref{s:states-steps} introduces the
states and the primitive steps used to construct these traces, and
Section~\ref{s:traces} defines the set of traces that are used to
represent commands,
along with some healthiness properties on sets of traces that represent commands.

\subsection{States and primitive steps}\label{s:states-steps}

To specify the behaviour of an executable command, 
it is sufficient to specify the allowable sequences of atomic steps $(\ssp)$ 
from state $\sigma \in \Sigma$ to state $\sigma' \in \Sigma$ made by the program,
where $\Sigma$ is the set of all possible program states.
The semantics presented below is largely independent of the form of the set of program states $\Sigma$.
One possible representation is as a store, i.e.\ a mapping from variables to values (including undefined), 
$\Sigma \sdefs Var \fun Val_{\bot}$,
which we use here for illustration.
Other representations are possible, for example, including a heap as well as a store.

Reactive-sequence semantics \cite{BoerHannemanDeRoever99,DeRoever01} uses traces of the form
\(
         \seqex{
              (\sigma_1, \sigma_2), 
              (\sigma_3, \sigma_4)
         },
\)
where the gap between $\sigma_2$ and $\sigma_3$ implicitly represents interference from the environment. 
In order to handle rely conditions, Peter
Aczel~\cite{Aczel83,CoJo07,BoerHannemanDeRoever99,DeRoever01} needed
to explicitly represent both environment and program steps,
differentiating an environment step from $\sigma$ to $\sigma'$,
written $\sigma\Estepd\sigma'$ here, from a program step, written
$\sigma\Pstepd\sigma'$.  When steps are combined to form a trace, the
post state $\sigma_1$ of each step $\sigma_0 \Pstepd \sigma_1$ must
equal the pre state of the next step $\sigma_1 \Pstepd \sigma_2$ --
and such a trace is said to be \emph{consistent}.  Because of this
consistency, the two steps in sequence can be abbreviated to $\sigma_0
\Pstepd \sigma_1 \Pstepd \sigma_2$.  For example, the following is a
valid \emph{Aczel trace}.
\begin{equation}\label{AczelTrace2}
	       \sigma_0 \Estepd \sigma_1 \Pstepd \sigma_2 \Estepd \sigma_3 \Pstepd \sigma_4 \Pstepd \sigma_5 \Estepd \sigma_6
\end{equation}
Note that environment steps are allowed before any program steps,
$\sigma_0 \Estepd \sigma_1$, and after all program steps,
$\sigma_5 \Estepd \sigma_6$.
Aczel traces also allow multiple program steps without intervening environment
steps and vice versa.%
\footnote{De Roever \cite{DeRoever01} makes use of labels unique to
  each location and each thread in a program, rather than program
  ($\pstepd$) and environment ($\estepd$) markers.  Here we follow
  Aczel's original simpler formulation more closely.}

Sequences of steps may be finite or infinite, and we use infinite
Aczel traces to represent non-terminating behaviours. Finite traces
represent behaviours that have either terminated successfully, aborted
due to a failure of the program, aborted due to a failure of the
environment, or become infeasible. These four cases are distinguished
by including one of four special kinds of atomic step at the end of
each finite trace. These special steps may only appear as the last
step in a trace.
\begin{itemize}
\item Like De Roever~\cite{DeRoever01}, we use steps of the form
  $\sigma \tmtd$ to represent termination from some initial state
  $\sigma$. For example the Aczel trace $\sigma_0 \Estepd \sigma_1
  \Pstepd \sigma_2 \tmtd$, performs two steps and then terminates from
  the last state $\sigma_2$.
\item To allow for aborting behaviour of programs, such as a
  precondition or rely condition failing to hold, the state space
  $\Sigma$ is extended with a distinguished undefined state
  $\botstate$, giving us $\Sigma_{\bot} \sdefs \Sigma | \botstate~.$
  We then use step $\sigma \Pstepd \botstate$ to represent the program
  taking an aborting step, and $\sigma\Estepd\botstate$ to represent
  the environment taking an aborting step, e.g.\ $\sigma_0 \Estepd
  \sigma_1 \Pstepd \sigma_2 \Pstepd \botstate$ or $\sigma_0 \Estepd
  \sigma_1 \Pstepd \sigma_2 \Estepd \botstate$.
\item A step from a state that does not lead anywhere, such as
    $\sigma$ on its own, represents a step that is
    infeasible, e.g.\ the trace $\sigma_0 \Estepd \sigma_1 \Pstepd
    \sigma_2$ performs two steps and then becomes infeasible from the
    last state $\sigma_2$. Programs may become infeasible right away,
    before taking any steps, e.g. the trace $\sigma_0$ is infeasible
    from its initial state $\sigma_0$. At first glance, the inclusion of
    infeasible traces may seem unnecessary, since they are, after all,
    unimplementable. However, in the same way that partial functions
    may be combined to create total functions, infeasible program
    behaviours may be combined to specify implementable ones, and the
    semantics benefits from the additional expressivity. For example,
    infeasibility may be used to represent the fact
    that a specification has conflicting requirements that cannot be
    met.
\end{itemize}

\subsection{Traces and commands}\label{s:traces}

We now formalise the Aczel traces that were introduced in the previous section.
Fundamental definitions for traces used throughout the paper are contained in
\reffig{summary}.
A trace -- see~(\ref{Tr}) in
Figure~\ref{fig:summary} -- is denoted by an initial state $\sigma \in
\Sigma$ and a possibly infinite sequence of atomic steps $t \in \seq
Step$, where we define
\begin{eqnarray*}
Step  & \sdefs &   \Pstep{\Sigma_{\bot}} | \Estep{\Sigma_{\bot}} | \tmtd ~.
\end{eqnarray*}
For example, the trace written informally as 
$ \Tra{\sigma_0}{~\Estepd ~ \sigma_1 ~ \Pstepd ~ \sigma_2 ~ \Estepd ~ \sigma_3 ~ \Pstepd ~ \sigma_4 ~ \Estepd ~ \sigma_5 ~ \tmtd} $
is formalised as the trace
\begin{displaymath}
  \Trf{\sigma_0}{\Estep{\sigma_1}, \Pstep{\sigma_2}, \Estep{\sigma_3}, \Pstep{\sigma_4}, \Estep{\sigma_5}, \tmtd}~.
\end{displaymath}
This form relies on Aczel traces being consistent. The $Step$ definition
does not need to include explicitly infeasible steps, since
these are implicitly represented. For example, the
trivial trace,
\(
  \Trf{\sigma}{~},
\)
represents infeasible behaviour starting from $\sigma$.%
\footnote{Here we differ from both Aczel \cite{Aczel83} and De Roever \cite{DeRoever01},
who do not allow empty traces.}
A healthiness condition on traces is that 
a termination step, $\tmtd$, or an aborting step, either $\Pstep{\botstate}$ or $\Estep{\botstate}$, can
occur only as the last step of a (finite) trace (\ref{Tseq}).

We define a partial ordering, $\leq$, on traces that describes the
conditions under which one trace $\Trf{\sigma_2}{t_2}$ is a valid
refinement of another $\Trf{\sigma_1}{t_1}$ as follows. 
\begin{eqnarray}
\Trf{\sigma_1}{t_1} \leq \Trf{\sigma_2}{t_2} & \sdefs &   \sigma_1 = \sigma_2 \land {} \nonumber\\
& &  (t_2 \in prefix(t_1) \lor (\exists t_3 \in Tseq \spot finite(t_3) \land t_1 = t_3 \cat \Pstep{\botstate} \land t_3 \in prefix(t_2)) \label{TRREF}
\end{eqnarray}
Informally, this captures the notion that a trace may be refined by
itself (since for any $t\in Tseq$, $t \in prefix(t)$), or by a trace
that replaces a program abort step by any subsequent behaviour, or
by a (prefix) trace that becomes infeasible.
For example we
have that
\begin{eqnarray*}
  && \Trf{\sigma_0}{\Pstep{\botstate}}\\
  \leq
  && \Trf{\sigma_0}{\Estep{\sigma_1}, \Pstep{\botstate}}\\
  \leq
  && \Trf{\sigma_0}{\Estep{\sigma_1}, \Pstep{\sigma_2}, \Pstep{\sigma_3}, \Estep{\sigma_4}, \tmtd}\\
  \leq
  && \Trf{\sigma_0}{\Estep{\sigma_1}, \Pstep{\sigma_2}}\\
  \leq
  && \Trf{\sigma_0}{} ~. 
\end{eqnarray*}
For each initial state $\sigma_0$, the set of traces that start from
$\sigma_0$ form a complete lattice under the trace refinement
ordering, with least element $\Trf{\sigma_0}{\Pstep{\botstate}}$, and
greatest element $\Trf{\sigma_0}{}$. In this ordering there is an
asymmetry to the interpretation of program abort steps and
environment abort steps. The appearance of a $\Pstepd{\botstate}$
step in a trace signifies that any behaviour whatsoever is acceptable
from that step onwards (hence it may be arbitrarily refined), whereas
$\Estepd{\botstate}$ does not: the only step tolerable at that point
is environmental failure. Intuitively, this reflects our desire to
distinguish an aborting program -- one that may exhibit any possible
behaviour -- from one that will not fail unless that failure is
brought about by its environment.

In our trace lattice for an initial $\sigma_0$, we have, for example, that
\begin{eqnarray*}
  \Trf{\sigma_0}{\Estep{\sigma_1}, \Pstep{\sigma_2}, \Pstep{\sigma_3}, \Estep{\sigma_4}, \tmtd}
  & \sqcup & 
  \Trf{\sigma_0}{\Estep{\sigma_1}, \Estep{\sigma_2}, \Pstep{\sigma_3}, \tmtd}
  ~=~ \Trf{\sigma_0}{\Estep{\sigma_1}}
\end{eqnarray*}
where 
$\sqcup$ is the 
supremum operator.
Note how the trace becomes infeasible at the first incompatible step.

The trace ordering is lifted pointwise on initial states to define an
ordering on the set of commands themselves.
The set of commands, (\ref{Com}) in Figure~\ref{fig:summary}, are
denoted by sets of traces with additional healthiness requirements
imposed by the $\prefixclosure{..}$ and $\abortclosure{..}$ conditions.
These conditions are equivalent to requiring that each command has at
least one trace $\Trf{\sigma}{~}$, starting from each possible initial
state $\sigma$, and that they are up-closed with respect to the trace
ordering (\ref{TRREF}). This allows us to define refinement of
commands by trace inclusion as per usual: for commands $c$ and $d$ we
have that
\begin{eqnarray*}
  c \refsto d & \sdefs & c \supseteq d~.
\end{eqnarray*}

In more detail, the empty closure (\ref{empty-closure}) of a set of
traces, $s$, adds all traces of the form $\Trf{\sigma}{~}$ for all
$\sigma \in \Sigma$ and the prefix closure (\ref{prefix-closure}) adds
all prefixes of the traces of $s$ as well.
The aborting traces of $s$ are those ending in $\Pstep{\botstate}$ but with that last step removed (\ref{aborting}).
Abort completing a set of (aborting) traces adds all possible extensions of the traces (\ref{abort-complete})
and the abort closure of a set of traces completes all its aborting traces (\ref{abort-closure}). 
Note that the abort closure only applies to traces ending in $\Pstep{\botstate}$ and not traces ending in $\Estep{\botstate}$ for the reasons discussed above.
Commands are represented by sets of traces that are both prefix closed (including empty closed) and abort closed (\ref{Com}).\footnote{From the perspective of mechanisation using a theorem prover, we find the definitions of prefix closure and abort closure more practical than up-closing directly with respect to the trace ordering (\ref{TRREF}).}

\begin{figure}
Let $s$ be a set of traces.
For a sequence $t$, $\dom(t)$ stands for the set of indices of $t$,
which may be infinite if $t$ is an infinite sequence.
\begin{eqnarray}
  Tr & \sdefs & \Sigma \cross Tseq
    \label{Tr} \\
  Tseq & \sdefs & \Zset{ t \in \seq Step }{ \forall i \in \dom(t) \spot t(i) \in \{ \Pstep{\botstate}, \Estep{\botstate}, \tmtd \} \implies i+1 \notin \dom(t)) }{} 
    \label{Tseq} \\
  \emptyclosure{s} & \sdefs & s \union \{ (\sigma,t) \in Tr | t = [\,] \} 
    \label{empty-closure} \\
  \prefixclosure{s} & \sdefs & \emptyclosure{s} \union \Zset{ \Trt{\sigma}{t} \in Tr }{(\exists t' \spot \Trt{\sigma}{t'} \in s \land t \in prefix(t')) }{} 
    \label{prefix-closure} \\
  aborting(s) & \sdefs & \{ \Trt{\sigma}{t} \in Tr | finite(t) \land \Trt{\sigma}{t \cat [\Pstep{\botstate}]} \in s\}
    \label{aborting} \\
  abort\_complete(s) & \sdefs & \Zset{ \Trt{\sigma}{t} \in Tr }{ (\exists t' \in prefix(t) \spot (\sigma,t') \in s) }{} 
    \label{abort-complete} \\
  \abortclosure{s} & \sdefs & s \union abort\_complete(aborting(s))
    \label{abort-closure} \\
  \Command & \sdefs & \Zset{ s \in \power Tr }{ s = \prefixclosure{s} \land s = \abortclosure{s} }{}
    \label{Com}
\end{eqnarray}

\caption{Traces and associated auxiliary functions}\label{fig:summary}
\end{figure}

\section{Primitive commands}\label{s:primitives}

A novelty of our approach is that the semantics of commands is defined in terms
of a small set of primitive commands 
plus a set combinators, such as sequential and parallel composition.
The formal definition of the primitive commands is given in Figure~\ref{fig:primitive-commands} in which 
the notation $\Zset{x \in S}{p(x)}{e(x)}$ stands for the set of all the values of the expression $e(x)$ 
for the bound variable $x$ ranging over the values in the set $S$ such that the predicate $p(x)$ holds,
$\Zset{x \in S}{}{e(x)}$ abbreviates $\Zset{x \in S}{\True}{e(x)}$, and
$\Zset{x \in S}{p(x)}{}$ abbreviates $\Zset{x \in S}{p(x)}{x}$.
Aczel's program and environment steps motivate the first two primitive commands,
and De Roever's terminating step motivates the third.
\begin{itemize}
\item
$\cpstepr$ can make a single program step $\sigma\Pstepd\sigma'$ satisfying binary relation $r$, 
i.e.\ $(\sigma,\sigma') \in r$, and terminate (\ref{cpstepr}).
\item
$\cestepr$ can make a single environment step $\sigma\Estepd\sigma'$ satisfying $r$ and terminate (\ref{cestepr}).
\item
$\cgdp$ can make a terminating step $\sigma\tmtd$ from any state $\sigma$ in the set of states $p$ 
(\ref{cgdp}).
\end{itemize}
The prefix and empty closures in definitions (\ref{cpstepr})--(\ref{cgdp}) ensure 
the commands satisfy the healthiness condition of being abort and prefix (and hence empty) closed.
Our final three primitive commands handle aborting behaviours.
\begin{itemize}
\item
$\Abort$ allows any behaviour at all (\ref{Abort}). It equals the
  abort and prefix-closure of the set $\Zset{ \sigma \in \Sigma }{}{
    \Trf{\sigma}{\Pstep{\botstate}}}$. 
\item $\EAbort$ allows the environment to abort (\ref{EAbort}).
\item $\cestepbotr$ allows the environment to either abort or do a step $\sigma\Estepd\sigma'$ satisfying $r$ and terminate (\ref{cestepbotr}).
\end{itemize}

\begin{figure}
Let 
$r$ be a binary relation on states, $r \subseteq \Sigma \cross \Sigma$,
and
$p$ be a set of states, $p \subseteq \Sigma$.
\begin{eqnarray}
  \cpstepr & \sdefs & \prefixclosure{\Zset{ (\ssp) \in r }{}{ \Trf{\sigma}{\Pstepsp, \tmtd} }} 
    \label{cpstepr} \\
  \cestepr & \sdefs & \prefixclosure{\Zset{ (\ssp) \in r }{}{ \Trf{\sigma}{\Estepsp, \tmtd} }} 
    \label{cestepr} \\
  \cgdp   & \sdefs & \emptyclosure{\Zset{ \sigma \in p }{}{ \Trf{\sigma}{\tmtd} }} 
  \label{cgdp} \\
  \Abort & \sdefs & Tr
    \label{Abort} \\
  \EAbort & \sdefs & \emptyclosure{\Zset{ \sigma \in \Sigma }{}{ \Trf{\sigma}{\Estep{\botstate}} }} 
    \label{EAbort} \\
      \cestepbotr & \sdefs & \cestepr \cup \EAbort 
    \label{cestepbotr}
\end{eqnarray}
\caption{Semantics of primitive commands}\label{fig:primitive-commands}
\end{figure}

We introduce the following abbreviations for the
common cases where the (program or environment step) relation
is the universal relation, $\universalrel \sdefs
\Sigma \cross \Sigma$, that relates all pairs of states, and where the
test set is either the set of all states, $\Sigma$, or the empty
set, $\emptyset$:
\begin{displaymath}
 \begin{array}[t]{rcl}
  \cpstepd & \sdefs & \cpstep{\universalrel} \\
  \Nil & \sdefs & \cgd{\Sigma}
 \end{array}
 \hspace*{2em}
 \begin{array}[t]{rcl}
  \cestepd & \sdefs & \cestep{\universalrel} \\
  \Magic & \sdefs & \cgd{\emptyset}
 \end{array}
 \hspace*{2em}
 \begin{array}[t]{rcl}
   \cestepbotd & \sdefs & \cestepbot{\universalrel}
 \end{array}
\end{displaymath}
The command $\cpstepd$ can make any single program step and terminate;
$\cestepd$ can make any single environment step and terminate, and
$\cestepbotd$ allows the environment to abort or take any environment
step and terminate.  Command $\Nil$ terminates immediately from any
initial state,
while $\Magic$ is infeasible from all initial states, having no traces
other than $\Trf{\sigma}{~}$ for all states $\sigma \in \Sigma$.
Note that $\cpstep{\emptyrel} = \cestep{\emptyrel} = \Magic$.

If $r_2 \subseteq r_1$, the refinements $\cpstep{r_1} \refsto
\cpstep{r_2}$ and $\cestep{r_1} \refsto \cestep{r_2}$ hold, and if
$p_2 \subseteq p_1$, $\cgd{p_1} \refsto \cgd{p_2}$.  Commands form a
complete lattice under the refinement ordering with $\Abort$ as the
least element and $\Magic$ as the greatest element, that is, for any
$c$, $\Abort \refsto c \refsto \Magic$.

\section{Program combinators}\label{s:combinators}

This section defines the primitive program combinators found in 
Figure~\ref{fig:primitive-operators} (using the auxiliary functions
defined in both Figures~\ref{fig:summary} and
\ref{fig:primitive-operators}).
The choice operators (the infimum and supremum of the program lattice) can
be found in Section~\ref{s:choice}, sequential composition in
Section~\ref{s:sequential}, and Section~\ref{s:iteration} uses these
to define iteration constructs using fixed points.
Section~\ref{s:parallel} defines parallel composition and
Section~\ref{s:strict-conjunction} a weak conjunction
operator~\cite{FACJexSEFM-14} that is used in Section~\ref{s:rely} to
define constructs suitable for specifying relies and guarantees.
Section~\ref{s:hide} introduces a program variable unrestriction
operator used in the definition of local variable blocks.

\subsection{Infimum and supremum}\label{s:choice}

A nondeterministic choice (or \emph{infimum}), $\Nondet C$, over a
set $C$ of commands gives the union of the behaviour of its components
(\ref{nondeterministic-choice}).  In this definition, the empty
closure ensures that the choice is a valid command, even if the choice
is over the empty set. Supremum ($\Supremum$) in the lattice of
commands is represented by intersection of sets of traces.

As binary operators $\sqcap$ (\ref{binary-choice}) and $\sqcup$
(\ref{binary-sup}) are associative, commutative and
idempotent. Operator $\sqcap$ has the identity $\Magic$ (i.e.\ $\Magic
\nondet c = c$) and annihilator $\Abort$ (i.e.\ $\Abort \nondet c =
\Abort$), and $\sqcup$ has the identity $\Abort$ and annihilator
$\Magic$.
When applied to primitive commands, the choice operators satisfy the following properties.
\begin{eqnarray*}
     \cgd{p_1} \nondet \cgd{p_2} ~~=~~  \cgd{p_1 \union p_2}
     &~~~\textrm{and}~~~&
     \cgd{p_1} \sqcup \cgd{p_2} ~~=~~  \cgd{p_1 \int p_2} \\
     \cpstep{r_1} \nondet \cpstep{r_2}  ~~=~~  \cpstep{r_1 \union r_2}
     &~~~\textrm{and}~~~&
     \cpstep{r_1} \sqcup \cpstep{r_2}  ~~=~~  \cpstep{r_1 \int r_2} \\
     \cestepbot{r_1} \nondet \cestepbot{r_2} ~~=~~ \cestepbot{r_1 \union r_2}
     &~~~\textrm{and}~~~&
     \cestepbot{r_1} \sqcup \cestepbot{r_2} ~~=~~ \cestepbot{r_1 \int r_2}
\end{eqnarray*}

\begin{figure}
Let $C$ be a set of commands, $c$, $c_1$ and $c_2$ be commands and $x$ a variable name.
\begin{eqnarray}
  \Nondet C & \sdefs & \emptyclosure{\Union C}
    \label{nondeterministic-choice} \\
  c_1 \nondet c_2 & \sdefs & \Nondet \{c_1 , c_2 \}   ~~=~~  c_1 \union c_2 
  \label{binary-choice} \\
  \Supremum C & \sdefs & (\Intersect C) 
    \label{supremum} \\
    c_1 \sqcup c_2 & \sdefs & \Supremum \{c_1 , c_2 \}   ~~=~~  c_1 \int c_2 
    \label{binary-sup} \\
    c_1 \Seq c_2 & \sdefs & \begin{array}[t]{l}
         abort\_closure(c_1 - terminating(c_1)) \cup (terminating(c_1) \cat c_2)
      \end{array}
    \label{sequential} \\
  c\FinIter & \sdefs & (\nu x \spot \Nil \nondet c \SSeq x)
    \label{finite-iteration} \\
  c\InfIter & \sdefs & (\mu x \spot \Nil \nondet c \SSeq x) 
    \label{finite-or-infinite-iteration} \\
  c^{\infty} & \sdefs & c\InfIter \SSeq \Magic
    \label{infinite-iteration} \\
  c_1 \parallel c_2 & \sdefs & 
                                                \abortclosure{\Zset{ tr \in Tr }{ 
                                                (\exists tr_1 \in c_1, tr_2 \in c_2 \spot match\_trace(tr_1,tr_2,tr)) }{} }
       \label{parallel} \\
  c_1 \together c_2 & \sdefs & (c_1 \int c_2) \union abort\_first(c_1,c_2) \union abort\_first(c_2,c_1)
    \label{weak-conjunction} \\
  c \hide x & \sdefs & \{(\sigma, t) \in Tr | \exists (\sigma', t') \in c \spot x \ndres (\sigma, t) = x \ndres (\sigma', t') \}
	\label{hide}
\end{eqnarray}
where for $s$, $s_1$ and $s_2$ sets of traces, $z_1$ and $z_2$ steps, $\sigma$, $\sigma_1$ and $\sigma_2$ states,
and $t$, $t_1$ and $t_2$ traces,
\begin{eqnarray}
  terminating(s) & \sdefs & \{ \Trt{\sigma}{t} \in Tr | finite(t) \land \Trt{\sigma}{t \cat [\tmtd]} \in s \}
    \label{terminating} \\
  s_1 \cat s_2 & \sdefs & \{ \begin{array}[t]{l} 
                                            (\sigma_1, t_1) \in s_1; (\sigma_2, t_2) \in s_2 | finite(t_1) \land  last\_state(\sigma_1,t_1) = \sigma_2 \spot \\
                                            ~~~~~ (\sigma, t_1 \cat t_2) \}
                                           \end{array}
    \label{concatentation} \\
  last\_state(\sigma,\seqex{})  & = & \sigma \\
  last\_state(\sigma, t \cat \seqex{z}) & = & \sigma' ~~~~\mbox{if $z \in \{ \Pstepsp, \Estepsp \}$ } 
    \label{last-state} \\
  match\_step(z_1,z_2,z) & \sdefs & \exists \sigma \in \Sigma_{\bot} \spot 
     \begin{array}[t]{l}
       z_1 = \Psteps \land z_2 = \Esteps \land z = \Psteps \lor {} \\
       z_1 = \Esteps \land z_2 = \Psteps \land z = \Psteps \lor {} \\
       z_1 = \Esteps \land z_2 = \Esteps \land z = \Esteps \lor {} \\
       z_1 = \tmtd \land z_2 = \tmtd \land z = \tmtd
     \end{array} 
     \label{match-step} \\
  match\_trace(\Trt{\sigma_1}{t_1},\Trt{\sigma_2}{t_2},\Trt{\sigma}{t}) & \sdefs & \sigma_1 = \sigma_2 = \sigma \land 
                                                          \dom(t_1) = \dom(t_2) = \dom(t) \land {} \\
                                        &  &           (\forall i \in \dom(t) \spot match\_step( t_1(i), t_2(i), t(i) ) ) 
       \label{match-trace} \\
  abort\_first(c_1,c_2) & \sdefs & abort\_complete(aborting(c_1) \int c_2)
    \label{abort-first} \\
  x \ndres (\sigma,t) & \sdefs & (\{x\} \ndres \sigma, (\lambda i \in \dom(t) \spot \{x\} \ndres t(i)))
\end{eqnarray}

\caption{Semantics of primitive operators}\label{fig:primitive-operators}
\end{figure}

\subsection{Sequential composition}\label{s:sequential}

Sequential composition $c_1 \Seq c_2$, written $c_1 c_2$
for short, retains all the unterminated traces of $c_1$ (these include
infeasible/incomplete (prefix) traces, aborting traces and infinite
traces) as well as the terminated traces $tr_1 \in c_1$ concatenated
with traces $tr_2 \in c_2$ that start from the last state of $tr_1$
(\ref{sequential}).  The definition makes use of the set of
terminating traces of a command (\ref{terminating}); note that the
final $``\tmtd$'' step is removed from the terminating traces as part
of this definition.  For two sets of traces $s_1$ and $s_2$, for which
the traces of $s_1$ are finite, $s_1 \cat s_2$ forms the set of
concatenations of traces $tr_1 \in s_1$ and $tr_2 \in s_2$ such that
the last state (\ref{last-state}) of $tr_1$ matches the initial state
of $tr_2$ (\ref{concatentation}).  The abort closure of the
unterminated traces of $c_1$ is used to ensure the set of traces of
the sequential composition is abort closed.

Sequential composition is associative.
It has identity $\Nil$ and
left annihilators of $\Magic$ and $\Abort$.
However note that in general, 
$(c \SSeq \Abort) \neq \Abort$ and $(c \SSeq \Magic) \neq \Magic$,
for example, if $c$ has no terminating traces.
Sequential composition distributes over nondeterministic choice from the right:
$(\Nondet C) \Seq d = \Nondet \{ c \in C \spot c \Seq d \}$
but only over finite choices from the left: 
$c \Seq (d_1 \nondet d_2) = (c \Seq d_0) \nondet (c \Seq d_2)$.
Sequential compositions of primitive commands satisfy the following properties,
where $p \dres r$ is the relation $r$ restricted so that its domain is contained in $p$.
\begin{eqnarray*}
  \cgd{p_1} \SSeq \cgd{p_2} & = & \cgd{p_1 \int p_2} \\
  \cgd{p} \SSeq \cpstep{r} & = & \cpstep{p \dres r} \\
  \cgd{p} \SSeq \cestepbot{r} & = & \cestepbot{p \dres r}
\end{eqnarray*}
Notationally we give sequential composition a higher precedence than all the other binary
operators.  

Using nondeterministic choice and sequential composition one can define 
the precondition command $\Pre{p}$ (\ref{precondition});
a command, $\optr{r}$, that allows a program step $\cpstep{r}$ to be optional for states $\sigma$ such that 
$(\sigma,\sigma) \in r$ by using a silent $\tau$ step instead in this case,
thus allowing code optimisations that remove redundant steps (\ref{optr});
and
an atomic update of a variable $x$ to a constant $\lvar$, that is defined using $opt$
so that it can be optimised to the null command if $x$ already equals $\lvar$ (\ref{update}).
The notation $\sigma[x \backslash \lvar]$ stands for the state $\sigma$
with the value of $x$ updated to be $\lvar$.
\begin{eqnarray}
  \Pre{p} & \sdefs & \cgdp \nondet (\cgd{\overline{p}} \SSeq \Abort)
    \label{precondition} \\
  \optr{r} & \sdefs & \cpstep{r} \nondet \cgd{\Zset{ \sigma \in \Sigma }{ (\sigma,\sigma) \in r }{} }
    \label{optr} \\
  update(x,\lvar) & \sdefs & \optr{\Zset{ (\sigma, \sigma') \in \Sigma \times \Sigma }{ \sigma' = \sigma[x \backslash \lvar] }{} }
	\label{update}
\end{eqnarray}
The precondition command $\Pre{p}$ terminates immediately from states in $p$ 
but aborts from states in the complement $\overline{p}$ of $p$.
Note that 
$\Pre{\Sigma} = \cgd{\Sigma} \nondet (\cgd{\emptyset} \SSeq \Abort) = \Nil \nondet (\Magic \SSeq \Abort) = \Nil$ 
and $\Pre{\emptyset} = \cgd{\emptyset} \nondet (\cgd{\Sigma} \SSeq \Abort) = \Magic \nondet (\Nil \SSeq \Abort) = \Abort$.

\subsection{Recursion and iteration}\label{s:iteration}

The set of all commands under the refinement ordering forms a complete lattice and
hence least and greatest fixed points can be defined over monotone functions from commands to commands.
For a monotonic function $f$ from commands to commands,
$\mu f$ is its least fixed point and $\nu f$ is its greatest fixed point.
As usual $\mu (\lambda x \spot f(x))$ is abbreviated as $\mu x \spot f(x)$,
and likewise for greatest fixed points.

The iteration operators are defined in terms of least ($\mu$) and greatest ($\nu$) fixed points
with respect to the refinement ordering.
The command 
$\Fin{c}$ iterates $c$ a finite number of times, zero or more (\ref{finite-iteration});
$\Om{c}$ iterates $c$ any number of times, zero or more and possibly infinitely many times (\ref{finite-or-infinite-iteration});
and
$\Inf{c}$ iterates $c$ an infinite number of times (\ref{infinite-iteration}).
The only feasible traces of $\Inf{c}$ are the infinite traces of $\Om{c}$.

A number of distinguished commands such as $\Skip$ and $\Idle$ can be
defined using iteration and the programming primitives.
Whereas $\Nil$ terminates immediately with no environment or program steps,
the command $\Skip$ allows any environment steps but no program steps,
and
the command $\Idle$ allows any environment steps as well as
a finite number of stuttering program steps 
(i.e.\ steps that do not change the state);
it represents the no-op command of code.%
\footnote{It is similar to the $\Idle$ command in the real-time context
that changes nothing but allows time to pass \cite{DaT,IWfSaRtPA}.}
The relation $\id$ is the identity relation on states.
\begin{eqnarray*}
 \begin{array}[t]{rcl}
  \Skip & \sdefs & \Om{\cestepbotd}
 \end{array}
 \hspace*{8em}
 \begin{array}[t]{rcl}
  \Idle & \sdefs & \Fin{(\cpstep{\id} \nondet \cestepbotd)} \SSeq \Om{\cestepbotd}
 \end{array}
\end{eqnarray*} 
Note that although $(\Nil \SSeq c) = c$,
it is not the case that $(\Skip \SSeq c) = c$ in general
because $\Skip$ allows environment steps that may change the state.
The command $\Idle$ is refined by $\Skip$ 
(and is equivalent to $\Skip$ modulo finite stuttering)
and $\Skip$ is refined by $\Nil$, i.e.\ $\Idle \refsto \Skip \refsto \Nil$.

Two further distinguished commands are $\Chaos$ and $\Term$.
The command $\Chaos$ represents a program that can do any non-aborting
behaviour and allows its environment to do anything, including abort.
The command $\Term$ represents terminating program behaviour (but
can't force its environment to terminate), i.e.\ if $\Term \refsto c$
then $c$ has finitely many program steps.  Termination is discussed in
more detail in Section \ref{s:termination}.
\begin{eqnarray*}
 \begin{array}[t]{rcl}
\Chaos & \sdefs & \Om{(\cpstepd \nondet \cestepbotd)}
 \end{array}
 \hspace*{5em}
 \begin{array}[t]{rcl}
\Term & \sdefs & \Fin{(\cpstepd \nondet \cestepbotd)} \SSeq \Om{\cestepbotd}
 \end{array}
\end{eqnarray*}

A preempted computation is one ending in an infinite sequence of environment steps.
Note that the command $\Term$ includes preempted behaviour ($\Term \refsto \Preempted$)
which reflects the notion that an implementation may not terminate if the environment takes over forever.
The command $\Forever$ represents all infinite behaviours.
Any non-aborting behaviour is either terminating or infinite: $\Chaos =
\Term \nondet \Forever$.  Such properties can be proved from the
algebraic properties of iterations \cite{Wright04}, such as $\Om{c} =
\Fin{c} \nondet \Inf{c}$ and $\Om{(c \nondet d)} = \Om{(\Om{d} \SSeq
  c)} \SSeq \Om{d}$.
\begin{eqnarray*}
 \begin{array}[t]{rcl}
  \Preempted & \sdefs & \Fin{(\cpstepd \nondet \cestepbotd)} \SSeq \Inf{\cestepbotd}
 \end{array}
 \hspace*{2em}
 \begin{array}[t]{rcl}
  \Forever & \sdefs & \Inf{(\cpstepd \nondet \cestepbotd)}
 \end{array}
\end{eqnarray*}

\subsection{Parallel composition}\label{s:parallel}

Parallel composition, $c_1 \pl c_2$,
synchronises a program step of $c_1$ with an environment step of $c_2$
to give a program step of $c_1 \pl c_2$ (and vice versa), and synchronises an
environment step of $c_1$ with an environment step of $c_2$ to give an
environment step of $c_1 \pl c_2$ -- see $match\_step$ (\ref{match-step}) in Figure~\ref{fig:primitive-operators}.
For example,
\[
  \begin{array} {ll}
    \Tra{\sigma_0}{~\Estepd~\sigma_1~\Pstepd~\sigma_2~\Pstepd~\sigma_3~\Estepd~\sigma_4~\tmtd} 
     & \mbox{ of $c_1$ in parallel with }\\
    \Tra{\sigma_0}{~\Pstepd~\sigma_1~\Estepd~\sigma_2~\Estepd~\sigma_3~\Estepd~\sigma_4~\tmtd}
     & \mbox{ of $c_2$ gives }\\
    \Tra{\sigma_0}{~\Pstepd~\sigma_1~\Pstepd~\sigma_2~\Pstepd~\sigma_3~\Estepd~\sigma_4~\tmtd}
     & \mbox{ of $c_1 \pl c_2$. }
  \end{array}
\]
Note that program steps can only be synchronised with a corresponding
environment step and cannot synchronise with another program step.
The synchronisations allow steps with after state $\botstate$ allowing synchronisation of 
aborting program and environment steps
or aborting environment steps of both programs.
For parallel composition a $\sigma\Pstepd\botstate$ step is matched by the
environment step $\sigma\Estepd\botstate$ to give the program step $\sigma\Pstepd\botstate$
of the parallel composition. 
Because a $\sigma\Pstepd\botstate$ step must be matched by a $\sigma\Estepd\botstate$ step, 
a (rather strong) form of specification is one that rules out $\sigma\Estepd\botstate$ steps. 
This gives a hint to the expressive power of the semantics.

Our approach differs from that of De Roever \cite{DeRoever01} in that 
we require both programs to terminate together,
whereas he allows one program to terminate early and the parallel composition to become the second program,
and hence he uses $\Nil$ as the identity of parallel composition.
Our approach is required to handle parallel composition of end-to-end specifications (see \refsect{end-to-end}).
Our definition of parallel composition does not include any notion of fairness
(see \refsect{progress}).
Parallel composition involving primitives satisfies simple rules such as the following.
\begin{eqnarray*}
  (\cgd{p_1} \SSeq c_1) \parallel (\cgd{p_2} \SSeq c_2) & = & \cgd{p_1 \int p_2} \SSeq (c_1 \parallel c_2)
  \\
  (\cpstep{r_1} \SSeq c_1) \parallel (\cestepbot{r_2} \SSeq c_2) & = & \cpstep{r_1 \int r_2} \SSeq (c_1 \parallel c_2) 
  \\
  (\cestepbot{r_1} \SSeq c_1) \parallel (\cestepbot{r_2} \SSeq c_2) & = & \cestepbot{r_1 \int r_2} \SSeq (c_1 \parallel c_2)
 \end{eqnarray*}

Parallel composition is associative and commutative.
Because $\Skip$ allows an arbitrary number of environment steps 
that may match any steps of $c$, $c \parallel \Skip = c$,
and hence $\Skip$ is the identity of parallel composition.
Note that $\Skip$ allows an infinite sequence of environment steps;
this allows it to synchronise with infinite traces of $c$.
Note that because $\Nil$ does not allow environment steps,
it is not the identity of parallel composition,
in fact $c \parallel \Nil$ can terminate only if $c$ behaves as $\Nil$.

A desirable property of parallel composition is that an infeasible program, e.g.\ $\Magic$,
in parallel with any command $c$, including $\Abort$, is infeasible.
In order to ensure $\Magic$ (which can make no steps at all)
annihilates $\Abort$, 
parallel composition requires that an aborting program step $\sigma \Pstepd \botstate$
is matched by an environment step $\sigma \Estepd \botstate$.
Because $\Magic$ doesn't include the step $\sigma \Estepd \botstate$ (it has no steps at all),
$\Magic \parallel \Abort = \Magic$.
More generally, parallel composition distributes over nondeterministic choice:
$(\Nondet C) \parallel d = \Nondet \{ c \in C \spot c \parallel d \}$.

\subsection{\Strictconjunction}\label{s:strict-conjunction}

Although the supremum operator allows us to conjoin commands,
i.e.\ $c_1 \sqcup c_2$ forms the intersection of the traces of $c_1$ and $c_2$
and hence is constrained to satisfy the specifications of both
$c_1$ and $c_2$, it is often too restrictive for our purposes. Most
programs fulfill commitments under given assumptions, and can fail if
those assumptions are not met. Often, one would like to constrain such
a specification, $c_1$ say, so that it satisfies an additional
commitment, $c_2$, only when $c_1$'s assumptions are met (i.e. it does
not abort). We use a \strictconjunction\ operator for that purpose,
and this allows us to handle rely and guarantee conditions is a
general way.%
\footnote{Some earlier papers refer to ``$\together$'' as strict conjunction (because it is abort strict).
This led to some confusion because it is not a strong as conjunction (intersection of traces).} 
The \strictconjunction\ of two commands,
$c_1 \together c_2$, synchronises non-aborting steps, e.g.,
\[\begin{array} {ll}
    \Tra{\sigma_0}{~\Estepd~\sigma_1~\Pstepd~\sigma_2~\Pstepd~\sigma_3~\tmtd}
     & \mbox{ of $c_1$ combines with}\\
    \Tra{\sigma_0}{~\Estepd~\sigma_1~\Pstepd~\sigma_2~\Pstepd~\sigma_3~\tmtd}
     & \mbox{ of $c_2$ to give}\\
    \Tra{\sigma_0}{~\Estepd~\sigma_1~\Pstepd~\sigma_2~\Pstepd~\sigma_3~\tmtd}
     & \mbox{ of $c_1 \together c_2$, }
\end{array}
\]
but if either performs a program abort step, it aborts, for example,
\[\begin{array} {ll}
    \Tra{\sigma_0}{~\Pstepd~\sigma_1~\Estepd~\sigma_2~\Pstepd~\botstate}
     & \mbox{ of $c_1$ combines with}\\
    \Tra{\sigma_0}{~\Pstepd~\sigma_1~\Estepd~\sigma_2}
     & \mbox{ of $c_2$ to give}\\
    \Tra{\sigma_0}{~\Pstepd~\sigma_1~\Estepd~\sigma_2~\Pstepd~\botstate}
     & \mbox{ of $c_1 \together c_2$. }
\end{array}
\]
The traces of $c_2$ include infeasible traces (and their prefixes)
so that in the \strictconjunction\ $c_1 \together c_2$, 
a trace of $c_1$ that leads to an aborting step is matched by a trace of $c_2$,
even if there are no feasible extensions of the trace of $c_2$.
A \strictconjunction\ $c_1 \together c_2$ contains the intersection of the traces of $c_1$ and $c_2$ as
well as the aborting traces of one for which there is a trace of the other that matches 
up until (but not including) the final abort step (\ref{weak-conjunction}).
As the sets of traces are prefix closed, any prefixes of $c_1$ and $c_2$ which match are included.
To handle \strictconjunction, care must be taken with commands that
include infeasible components. 
In the sequential refinement calculus for any terminating command $c$, 
\begin{eqnarray}\label{magicseqref}
                   c \SSeq \Magic = \Magic.
\end{eqnarray}
This property does not hold for the semantics given here because 
\begin{equation}\label{magic-together-abort}
 (c \SSeq \Abort) \together (c \SSeq \Magic) \refsto c \SSeq (\Abort \together \Magic) = c \SSeq \Abort 
\end{equation}
Were we to allow (\ref{magicseqref}) the left side reduces to 
\begin{eqnarray*}
 c \SSeq \Abort \together \Magic
\end{eqnarray*}
which assuming $c$ is not $\Abort$ reduces to $\Magic$.

\Strictconjunction\ is associative, commutative and idempotent.
It has identity $\Chaos$ and
$\Abort$ is an annihilator.
For example, one can deduce the following.
\begin{displaymath}
 \begin{array}{rcl}
  (\cgd{p_1} \SSeq c_1) \together (\cgd{p_2} \SSeq c_2) & = & \cgd{p_1 \int p_2} \SSeq (c_1 \together c_2) \\
  (\cpstep{r_1} \SSeq c_1) \together (\cpstep{r_2} \SSeq c_2) & = &  \cpstep{r_1 \int r_2} \SSeq (c_1 \together c_2) \\
  (\cestepbot{r_1} \SSeq c_1) \together (\cestepbot{r_2} \SSeq c_2) & = &  \cestepbot{r_1 \int r_2} \SSeq (c_1 \together c_2)
 \end{array}
 \hspace*{2ex}
 \begin{array}{rcl}
  c \together \Abort & = & \Abort \\
  (\cpstepr \SSeq c_0) \together (\cestepbotd \SSeq c_1) & = & \Magic
 \end{array}
\end{displaymath}

\subsection{Unrestricting a variable}\label{s:hide}

In order to define local variable blocks 
(in Section~\ref{s:local-variables})
the operator $c \hide x$ is used to define a program with traces that are the same as $c$, 
except they do not need to agree on values of the variable $x$.
Foster et al. use a similar operator in their work \cite[Sect.~5.4]{DBLP:conf/utp/FosterZW14}.
Let $x \ndres (\sigma,t)$ be the trace $(\sigma,t)$ with variable $x$ removed from $\sigma$ and 
from the state of every program and environment step in $t$.
The traces of the command $c \hide x$ are the same as those of $c$,
except they do not restrict the variable $x$ at all (\ref{hide}).

\section{A wide-spectrum language}\label{s:wide-spectrum}

In this section we extend the basic language to an imperative concurrent programming language
augmented with specification constructs. 
Section~\ref{s:expressions} addresses the semantics of expression evaluation 
without the common assumption that expression evaluation is atomic.
Section~\ref{s:imperative} make use of expressions to define typical imperative programming constructs:
assignment, ``if'' statement and ``while'' statement.
Section~\ref{s:termination} addresses specifying termination and fairness constraints.
End-to-end specifications are introduced in Section~\ref{s:end-to-end};
they are the generalisation of postcondition specifications to the concurrent context.
Section~\ref{s:rely} introduces rely and guarantee constructs that allow the language
to express Jones-style relational rely and guarantee specifications.
Local variable blocks are illustrated in Section~\ref{s:local-variables}.

\subsection{Expression evaluation}\label{s:expressions}

In the context of concurrency, expression evaluation is subject to interference on shared variables
and hence cannot be treated as being atomic.
As a consequence of this, expression evaluation is nondeterministic
because the interference may modify variables during evaluation.
Atomic read and write of shared variables is however assumed;
(structured) variables that cannot be read/written atomically need to be treated separately.

The syntax of expressions is standard, encompassing values, variables, and 
the typical unary and binary operators 
(to save space we do not consider generalisation to $n$-ary operators).
Although side effects within expressions can be handled by the semantics,
only expressions free from side effects are considered here
and hence any program steps made during expression evaluation are stuttering steps,
i.e.\ they do not modify the observable state.
To allow for any possible implementation strategy,
the set of traces of an expression evaluation is closed under finite stuttering,
i.e.\ if it includes a trace $t_1$,
it includes any other trace $t_2$ that is equivalent to $t_1$ modulo finite stuttering of program steps.
To allow for undefined expressions, the set $Val$ is augmented with the undefined value $\botvalue$
to give the set $Val_{\bot}$.

The command $\ceval{e}{\lvar}$ represents 
the evaluation of the expression $e$ to the value $\lvar \in Val_{\bot}$.
Evaluation of $e$ to $\lvar$ may either succeed or fail.
To model this in the semantics an evaluation that succeeds is terminating 
whereas a failing evaluation is infeasible and hence is eliminated from a nondeterministic choice
with a succeeding evaluation to some other value.
The semantics of expression evaluation is defined below,
where 
$x$ is a variable, $\ominus$ is a unary operator and $\oplus$ is a binary operator.
\begin{eqnarray}
  \ceval{\lvar_1}{\lvar_2} & \sdefs &\Idle \SSeq \cgd{\Zset{\sigma \in \Sigma}{\lvar_1 = \lvar_2}{}}  \qquad \qquad \qquad ~~\mbox{for $\lvar_1,\lvar_2 \in Val_{\bot}$} \label{expr-eval-lvar} \\
  \ceval{x}{\lvar}  & \sdefs & \Idle \SSeq \cgd{\{ \sigma \in \Sigma | \sigma(x) = \lvar \}} \SSeq \Idle ~~\qquad \qquad \mbox{for $\lvar \in Val_{\bot}$}   \label{expr-eval-var} \\
  \ceval{\ominus~e}{\lvar}  & \sdefs & \Nondet \Zset{ \lvar' \in Val_{\bot}}{ \lvar = eval(\ominus,\lvar') }{ \ceval{e}{\lvar'} }     \label{expr-eval-unary} \\
  \ceval{e_1~\oplus~e_2}{\lvar}  & \sdefs & \Nondet \Zset{ \lvar_1, \lvar_2 \in Val_{\bot} }{ \lvar = eval(\oplus,\lvar_1,\lvar_2) }{ (\ceval{e_1}{\lvar_1} \parallel \ceval{e_2}{\lvar_2}) }     \label{expr-eval-binary}
\end{eqnarray}
The evaluation $\ceval{\lvar_1}{\lvar_2}$ of a constant $\lvar_1$ to another constant $\lvar_2$ is represented by 
the command $\Idle$, which allows arbitrary interference and finite stuttering steps, followed by a test that the constants are equal (\ref{expr-eval-lvar}).
Note that the refinement $\ceval{\lvar_1}{\lvar_1} \refsto \ceval{\lvar_1}{\lvar_2}$ holds and hence 
$\Nondet \Zset{ \lvar_2 \in Val_{\bot} }{}{ \ceval{\lvar_1}{\lvar_2} } = \ceval{\lvar_1}{\lvar_1} = \Idle$.
Also, because constants $\lvar_1$ and $\lvar_2$ are independent of the state, 
the test $\cgd{\Zset{\sigma \in \Sigma}{\lvar_1 = \lvar_2}{}}$ is not subject to interference from the environment, 
and so $\ceval{\lvar_1}{\lvar_2} = \Idle \SSeq \cgd{\Zset{\sigma \in \Sigma}{\lvar_1 = \lvar_2}{}} \SSeq \Idle$.

To allow an assignment such as $x \asgn x$ to be optimised to $\Skip$ 
(which has no program steps),
the evaluation of the expression $x$ (and the assignment) must allow traces with no program steps.
The core of the definition of $\ceval{x}{\lvar}$ is therefore the instantaneous test 
$\cgd{\{ \sigma \in \Sigma | \sigma(x) = \lvar\}}$
which
is feasible only in states $\sigma$ such that $\sigma(x) = \lvar$.
A single reference to $x$ is assumed to be atomic; if $x$ was a
structured variable with non-atomic access, a more complex definition
would be needed.

Evaluation of a binary expression $e_1 \oplus e_2$ is nondeterministic 
because its value depends on the possibly changing values of variables
in $e_1$ and $e_2$ (and similarly for unary expressions).
For example, the traces of $\ceval{x + y}{\lvar}$
consist of all possible evaluations of both $\ceval{x}{\lvar_1}$ and $\ceval{y}{\lvar_2}$
such that $\lvar = \lvar_1 + \lvar_2$,
which is represented as a nondeterministic choice.
\begin{eqnarray*}
  \ceval{x + y}{\lvar} & = & \Nondet \Zset{ \lvar_1, \lvar_2 } { \lvar = \lvar_1 + \lvar_2 }{ \ceval{x}{\lvar_1} \parallel \ceval{y}{\lvar_2} }
\end{eqnarray*}
Subexpressions are evaluated in parallel to represent arbitrary order of
subexpression evaluation (including parallel evaluation).
For example, the evaluation of a (non-conditional) ``and'' of $b_1$ and $b_2$ to true
corresponds to evaluating both $b_1$ and $b_2$ to true in parallel, 
i.e.\ $ \ceval{b_1 \land b_2}{\True} = \ceval{b_1}{\True} \parallel \ceval{b_2}{\True}$.
\[
  \ceval{b_1 \land b_2}{\True}
 \Equals
   \Nondet \Zset{ \lvar_1, \lvar_2 }{ \True = eval(\land,\lvar_1,\lvar_2) }{ \ceval{b_1}{\lvar_1} \parallel \ceval{b_2}{\lvar_2} } 
 \Equals
   \Nondet \Zset{ \lvar_1, \lvar_2 }{ \lvar_1 = \True \land \lvar_2 = \True }{ \ceval{b_1}{\lvar_1} \parallel \ceval{b_2}{\lvar_2} }
 \Equals
   \ceval{b_1}{\True} \parallel \ceval{b_2}{\True}
\]
For undefined expressions, such as a division by zero, the result of the evaluation is the undefined value $\botvalue$,
for example, $\ceval{1/0}{\botvalue}$ succeeds.

\subsection{Imperative programming constructs}\label{s:imperative}

An assignment, $x \asgn e$, evaluates $e$ to a value $\lvar$
and then updates $x$ to $\lvar$ (\ref{assignment}).
While evaluation of the expression $e$ is non-atomic,
we assume that the act of storing of the value in the variable $x$ is atomic.
If the evaluation is undefined, the assignment aborts.
If the test in a conditional (\ref{conditional}) evaluates to true, $c_1$ is executed,
if it evaluates to false, $c_2$ is executed,
and
if its evaluation is undefined it aborts.
The body of a while loop (\ref{while-loop}) is repeated while its test evaluates to true.
The additional $(\cpstep{\id} \Seq \Idle)$ ensures it takes at least one program step on each iteration.
If its test evaluates to false, it terminates,
otherwise, if its test evaluation is undefined, it aborts.
\begin{eqnarray}
	x \asgn e
	&\sdefs&
	\Nondet \Zset{ \lvar \in Val }{}{
		\ceval{e}{\lvar} \SSeq update(x,\lvar) \SSeq \Idle} \nondet (\ceval{e}{\botvalue} \SSeq \Abort)
	\label{assignment}
	\\
  \If b \Then c_1 \Else c_2 & \sdefs & (\ceval{b}{\True} \SSeq c_1) \nondet (\ceval{b}{\False} \SSeq c_2) \nondet (\ceval{b}{\botvalue} \SSeq \Abort)
    \label{conditional} \\
  \While b \Do c & \sdefs & (\ceval{b}{\True} \SSeq c \Seq \cpstep{\id} \Seq \Idle)\InfIter \SSeq (\ceval{b}{\False} \nondet (\ceval{b}{\botvalue} \SSeq \Abort))
    \label{while-loop}
\end{eqnarray}

For example, the expansion of a simple assignment gives the following.
For brevity we assume $\lvar$ ranges over natural numbers only, but
the full case does not change the reasoning.

\begin{derivation}
	\step{
		x \asgn 1
	}

	\trans{=}{assignment definition}
	\step{
		\Nondet \Zset{ \lvar \in Val }{}{
		\ceval{1}{\lvar} \SSeq update(x,\lvar) \SSeq \Idle} \nondet (\ceval{1}{\botvalue} \SSeq \Abort)
	}

	\trans{=}{split nondeterministic choice}
	\step{
          (\ceval{1}{1}\SSeq update(x, 1) \SSeq \Idle) \nondet 
          (\Nondet \Zset{ \lvar \in Val }{\lvar \neq 1}{\ceval{1}{\lvar} \SSeq update(x,\lvar) \SSeq \Idle})
          \nondet (\ceval{1}{\botvalue} \SSeq \Abort)
	}

        \trans{=}{expand and simplify expression evaluation, $\Magic$ is left annihilator }
	\step{
          (\Idle \SSeq update(x, 1) \SSeq \Idle)
          \nondet ( \Nondet \Zset{ \lvar \in Val }{\lvar \neq 1}{\Idle \SSeq \Magic})
          \nondet (\Idle \SSeq \Magic)
	}
        
        \trans{=}{monotonicity of sequential composition, $\Magic$ is the greatest element}
        \step{
          \Idle \SSeq update(x, 1) \SSeq \Idle
	}
\end{derivation}

\subsection{Termination and fairness}\label{s:termination}

In the context of concurrency, termination of even the most obviously terminating program, 
such as ``$x := 0$'',
relies on the program being scheduled long enough to complete.
If the program is interrupted by its environment forever, it does not terminate.
For example, the program
\begin{displaymath}
  x \mathbin{:=} 1 \Seq ((\While x \neq 0 \Do \Skip) \parallel x \mathbin{:=} 0)
\end{displaymath}
will not terminate unless the right-hand operand in the parallel composition is given a chance to set $x$ to zero.

The command $\Term$ is defined as $(\cpstepd \nondet \cestepbotd)\FinIter \SSeq \cestepbotd\InfIter$
which, although it only allows a finite number of program steps, 
does not preclude its environment taking over forever.
Hence $\Term \refsto c$ is a slightly weak form of termination 
because it only requires termination of $c$ provided its environment doesn't interrupt it forever.

Fair execution of a program rules out its environment interrupting it forever.
The program that allows all behaviours except being interrupted forever can be expressed as
\begin{eqnarray*}
  \Fair & \sdefs & \cestepbotd\FinIter \SSeq (\cpstepd \SSeq \cestepbotd\FinIter)\InfIter
\end{eqnarray*}
and fair execution of a command $c$ can be represented by
\(
  \Fair \together c, 
\)
where \strictconjunction\ ensures any aborting behaviour of $c$ is preserved.
A stronger form of termination can be expressed as
\begin{eqnarray*}
  \Fairterm & \sdefs & (\cpstepd \nondet \cestepbotd)\FinIter ~~=~~ \cestepbotd\FinIter \SSeq (\cpstepd \SSeq \cestepbotd\FinIter)\FinIter
\end{eqnarray*}
and requires a finite number of both program and environment steps.
Interestingly,
\begin{eqnarray*}
  \Fairterm & = & \Fair \together \Term
\end{eqnarray*}
and hence if $\Term \refsto c$, by monotonicity,
\(
  \Fair \together \Term \refsto \Fair \together c,
\)
that is,
\(
  \Fairterm \refsto \Fair \together c,
\)
and hence weak termination of $c$ ensures strong termination if the execution of $c$ is fair.

\subsection{End-to-end specifications}\label{s:end-to-end}

In the rely-guarantee paradigm
postconditions are end-to-end and allow only a finite number of program steps
\cite{Jones81d,Jones83b}.  
A specification command, $\Sspec{}{}{q}$, satisfies a binary relation $q$ between its initial and final states \cite{TSS}.
\begin{eqnarray}
  \Sspec{}{}{q} & \sdefs & \Nondet \Zset{ \sigma \in \Sigma }{}{ (\cgd{\{\sigma\}} \SSeq \Term \SSeq \cgd{\{\sigma' | (\sigma,\sigma') \in q\}}) }
	\label{defn-spost}
\end{eqnarray}
It may take any sequence of steps that collectively establish $q$.  
Sequences of steps that do not satisfy $q$ are infeasible because of the final $\cgdd$ step.
Note that this specification construct contains all possible implementations.
In the literature it is often the case
that $q$ would initially be an atomic step, and this atomicity is then
successively ``split'' into an implementation
\cite{BackAtomicityRefinement}.
Brookes \cite{Brookes-full-abstraction} and Dingel \cite{Dingel02}
use a notion of program equivalence modulo finite stuttering and mumbling.%
\footnote{A trace $t_1$ is mumbling equivalent to a trace $t_2$ if 
$t_1$ contains a sub-trace $\Tra{\sigma_i}{~\Pstepd~\sigma_{i+1}~\Pstepd~\sigma_{i+2}}$
and
$t_2$ is the same as $t_1$ but with the sub-trace replaced by $\Tra{\sigma_i}{~\Pstepd~\sigma_{i+2}}$,
i.e.\ two consecutive program steps can be replaced by a single step with the same overall effect.}
The set of traces of a specification command is closed under finite stuttering and mumbling,
and hence we can safely use the stronger notion of program equivalence as equality on sets of traces.
For example, it is straightforward to show that a specification is closed under finite stuttering
because $\Sspec{}{}{q} \parallel \Idle = \Sspec{}{}{q}$.
Furthermore, (\ref{defn-spost}) allows a trace with no program steps 
for initial states $\sigma$, such that $(\sigma,\sigma) \in q$,
which isn't catered for by mumbling equivalence,
however, it would be straightforward to generalise mumbling equivalence to allow this.

End-to-end specifications as standalone commands are in general unimplementable.
The following program specifies the push of a value $v$ onto a stack $s$.
Here relations are expressed in predicative style, in which before state
values of a variable $s$ are represented by $s$ and final state values by
$s'$ (as in Z \cite{Spivey92,Hayes93,Woodcock96} or TLA \cite{Lamport03}).
We let $\seqT{v}$ represent the singleton sequence containing $v$.
\[
	\Post{s' = \seqT{v} \cat s}
\]
There are no refinements of this specification to code because code
does not constrain the environment to prevent modifications to $s$ after it
finishes its program steps.  To refine this specification to code requires
an assumption about the environment, in particular, that the environment
does not modify $s$.  In addition, this specification says nothing about
which (other) variables may be modified.  These parts of rely-guarantee
reasoning are introduced in the next section.

\subsection{Rely and guarantee constructs}\label{s:rely}\label{s:guarantee}

Because the model explicitly includes both environment and program steps,
defining specification constructs to handle rely and guarantee conditions is straightforward.

\subsubsection{Restricting the program: guarantees}

The command $\pguardg$ restricts its program steps to satisfy the relation $g$ 
and leaves its environment steps unconstrained.  
When combined with a command $c$ using \strictconjunction, $\pguard{g} \together c$,
it allows only program steps of $c$ that respect $g$
(but does not prevent the combination from aborting if $c$ can abort
after a sequence of program steps satisfying $g$).
This command encodes the \emph{guarantee} condition of Jones~\cite{Jones83b}.
\begin{eqnarray}
	\pguard{g}
	&\sdef&
	\Om{(\cpstep{g} \choice \cestepbotd)}
\end{eqnarray}
Guarantees satisfy useful properties, such as the following.
\begin{eqnarray*}
  \pguard{g_1} \together \pguard{g_2} & = & \pguard{g_1 \int g_2} \\
  \pguard{g_1} \parallel \pguard{g_2} & = & \pguard{g_1 \union g_2}   
\end{eqnarray*}
A frame $x$ on a command $c$, written $\Frame{x}{c}$,
constrains $c$ to only modify the set of variables $x$;
it can be defined as a guarantee to never modify variables outside
$x$.
The relation $\id(x)$ is the identity relation on the set of variables $x$ 
(i.e.\ the relation $\{ (\sigma,\sigma') | x \dres \sigma = x \dres \sigma' \}$,
where $x \dres \sigma$ is the state $\sigma$ with its domain restricted to $x$)
and $\overline{x}$ is the complement of the set of variables $x$.
\begin{eqnarray}
	\Frame{x}{c} & \sdefs & \pguard{\idbarx} \together c
\end{eqnarray}

As an example, recall pushing a value $v$ onto a stack $s$.  Typically an implementation is expected to modify $s$ and only $s$.  
\[
	\Post{s' = \seqT{v} \cat s}
	\together 
	\pguard{\id(\bar{s})}
\]
which is equal to the abbreviation $\Frame{s}{
	\Post{s' = \seqT{v} \cat s}
}$.
Now any implementation of this specification may modify $s$ only (in
addition to any local variables, described below).  However, it is still
infeasible for the same reasons that a stand-alone end-to-end specification
is infeasible; this problem will be rectified by the introduction of assumptions
about the environment, described below.

Note that the program $\Frame{s}{\Post{s' = \seqT{v} \cat s}}$, is not the same as
\begin{math}
	\Post{s' = \seqT{v} \cat s
	\land
	\id(\bar{s})}
\end{math},
because the latter may have intermediate program steps that modify variables other than $s$. 
Also, the latter specification constrains the environment steps in a way that the first command does not. 
In the first command, the environment may arbitrarily modify variables other than $s$, 
allowing the program to terminate in a state where variables other than $s$ have been changed by the environment. 
For the second program, every trace is constrained to ensure that all variables other than $s$ 
have a final value equal to their initial value,
although it makes no constraints on any of the intermediate values of these variables.
This difference is one of the most prominent between specification/refinement in sequential and parallel contexts.

\subsubsection{Assumptions about the environment: rely}

A command $\eguard{r}$ restricts every
environment step to satisfy the relation $r$ or abort,
while allowing any program steps.
A command $\envc{r}$ aborts if the environment does not respect $r$.
The commands
$\eguard{r}$
and 
$\envcr$
for environment steps
perform similar roles in the language to those of guards and assertions on
(pre-)states in
the sequential refinement calculus.

The intention is for $\envcr$ to be composed with other
specification commands using \strictconjunction:
any $d$ such that $c \together \envcr \refsto d$ need refine $c$ only
when the environment respects $r$, as under any other circumstance the
behaviour of $d$ is unconstrained.
This command encodes the \emph{rely} condition of Jones~\cite{Jones83b}.
\begin{eqnarray}
	\eguard{r}
	&\sdef&
	\Om{(\cpstepd \choice \cestepbotr)}
	\\
	\envc{r}
	&\sdef&
	\eguard{r} \SSeq (\Nil \choice \cestep{\notr} \SSeq \Abort)
\end{eqnarray}
Environment guards and assumptions satisfy the following properties.
\begin{eqnarray*}
  \eguard{r_1} \together \eguard{r_2} & = & \eguard{r_1 \int r_2} \\
  \envc{r_1} \together \envc{r_2} & = & \envc{r_1 \int r_2}
\end{eqnarray*}

\newcommand{\tunion}{~\union~}

\newcommand{\xdinc}{p}
\newcommand{\xpp}{x \scalebox{0.9}{\sf{++}}}

More concretely, assume command $\xdinc$ is a sequence of two atomic
steps that initialises $x$ to 0 and then increments $x$, for
instance, in a language with atomic initialisation and increments, $x
\asgn 0 \Seq \xpp$.  The traces of this command include all
terminated and non-terminated traces with exactly two program steps
interleaved with arbitrary environment steps:
\begin{displaymath}
  \xdinc ~~=~~
  \cestepbotd\InfIter \SSeq \cpstep{x'=0} \SSeq
  \cestepbotd\InfIter \SSeq \cpstep{x'=x+1} \SSeq
  \cestepbotd\InfIter~.
\end{displaymath}
In general we may not state anything about the final value of $x$ in
the terminating traces of $\xdinc$ since the environment steps may
modify $x$.  If we guard the environment to restrict to steps that do
not modify $x$, however, we obtain the following types of traces.
\begin{displaymath}
  \xdinc \together \eguard{x' = x} ~~=~~ 
  \cestepbot{x'=x}\InfIter \SSeq \cpstep{x'=0} \SSeq
  \cestepbot{x'=x}\InfIter \SSeq \cpstep{x'=x+1} \SSeq
  \cestepbot{x'=x}\InfIter
\end{displaymath}
In this new program, every environment step leaves $x$ unmodified or
aborts, and so we can deduce that $x$ equals $1$ in the final state of
any terminating behaviour.

In practice, restricting the environment is too strong for
specifications.  The use of a rely command, however, allows all
possible behaviours of the environment.
\begin{eqnarray*}
  \xdinc \together \envc{x' = x} & ~~=~~ &  
  (\cestepbot{x'=x}\InfIter \SSeq \cpstep{x'=0} \SSeq
  \cestepbot{x'=x}\InfIter \SSeq \cpstep{x'=x+1} \SSeq
  \cestepbot{x'=x}\InfIter
  )
  \sqcap {}\\
  && 
  (\cestepbot{x'=x}\InfIter \SSeq \cestepbot{x'\neq x} \SSeq \Abort)
  \sqcap {}\\
  && 
  (\cestepbot{x'=x}\InfIter \SSeq \cpstep{x'=0} \SSeq
  \cestepbot{x'=x}\InfIter \SSeq \cestepbot{x'\neq x} \SSeq \Abort)
  \sqcap {}\\
  &&
  (\cestepbot{x'=x}\InfIter \SSeq \cpstep{x'=0} \SSeq
  \cestepbot{x'=x}\InfIter \SSeq \cpstep{x'=x+1} \SSeq
  \cestepbot{x'=x}\InfIter \SSeq \cestepbot{x'\neq x} \SSeq \Abort)
\end{eqnarray*}
The traces of $\xdinc \together \envc{x' = x}$ therefore contain the
traces where the environment does not modify $x$, establishing $x$
equals $1$ in traces that terminate, as well as traces where the
environment does not respect the rely condition, at which point any
behaviour is possible.
When used in a specification, the rely command therefore absolves the
implementation of establishing anything when the condition is not
respected by its environment; similarly if a sequential program begins execution in a
state that does not satisfy its precondition it need not establish the
postcondition.

We can complete the specification of push onto a stack
in an environment
that does not modify $s$ by adding the appropriate rely.
\begin{equation}
\label{seq-push}
	s: \Post{s' = \seqT{v} \cat s}
	\together 
	\envc{\id(s)}
\end{equation}
This is now feasible: it may be implemented by the typical two-line linked
list code (after data refining the representation).
Traces of the implementation in which the stack is not
modified by its environment must satisfy the end-to-end relation $s' = \seqT{v} \cat s$, while
traces of the implementation in which the stack is modified by the environment are always
legal, since the specification aborts under such circumstances.

The rely $\id(s)$ is a strong assumption, essentially requiring that $s$
is local to the program (although its value may be read by the
environment).  It is often the case that the program and environment
communicate via shared variables, in which case a weaker rely is used.
For instance, if $i$ is a shared index into an array for concurrent search,
an implementation may allow each program to assume that $i$ only increases,
never decreases (given by the command $\envc{i' \geq i}$).
The weakest assumption is $\envc{\emptyrel}$, making no assumptions about
what the environment may modify.  This is usually too weak for feasible
implementations, for instance, $\Post{q} \together \envc{\emptyrel} =
\Post{q}$.
More fully developed refinements in a rely-guarantee context appear in
\cite[Sect 2.3]{FACJexSEFM-14} and \cite[Sect. 11]{HayesJonesColvin14TR}.

Traditionally rely-guarantee reasoning is formulated
as a quintuple statement
$\quintprgqc$, which states that, assuming that the initial state
satisfies $p$ and each step of the environment
satisfies $r$, then $c$ terminates and establishes $q$ and furthermore
every program step of $c$ satisfies the guarantee $g$ \cite{Jones83b}.
A rely-guarantee quintuple can be rewritten as a refinement statement in our language as follows.
\begin{eqnarray}
\label{rgquint-refsto}
		\assertp \SPostq \together \envcr \together \pguard{g}
		\refsto c
\end{eqnarray}
Note the separation of concerns 
into pre, post, rely and guarantee conditions
allowing separate reasoning about each component.

\subsubsection{Further properties of environments}

Because environment guards and relies are novel to our approach
we explore some useful properties of these commands.
\begin{lemma}[eguard-env]\label{eguard-env}
For a binary relation $r$,~~~
\(
  \eguard{r} \together \envc{r} = \eguard{r}~.
\)
\end{lemma}

\begin{proof}
The proof expands the definition of $\envc{r}$ and simplifies.
\[
  \eguard{r} \together \envc{r}
 \Equals
  \eguard{r} \together (\eguard{r} \nondet \eguard{r} \SSeq \cestep{\bar{r}} \SSeq \Abort)
 \Equals
  (\eguard{r} \together \eguard{r}) \nondet (\eguard{r} \together \eguard{r} \SSeq \cestep{\bar{r}} \SSeq \Abort)
 \Equals
  \eguard{r} \nondet (\eguard{r} \together \eguard{r} \SSeq \cestep{\bar{r}} \SSeq \Abort)
 \Equals
  \eguard{r}
\]
For the last step one needs to show 
$\eguard{r} ~\refsto~ (\eguard{r} \together \eguard{r} \SSeq \cestep{\bar{r}} \SSeq \Abort)$,
which can be shown by induction.
\end{proof}

\begin{lemma}\label{eguardr-refsto}
For a binary relation $r$, and commands $c$ and $d$,~~
\begin{eqnarray*}
  \eguard{r} \together c \refsto \eguard{r} \together d & ~~\iff~~ & c \refsto \eguard{r} \together d~.
\end{eqnarray*}
\end{lemma}

\begin{proof}
Assuming the left side, the right follows by the reasoning below because $\Chaos \refsto \eguard{r}$.
\[
  c ~=~ \Chaos \together c ~\refsto~ \eguard{r} \together c ~\refsto~ \eguard{r} \together d
\]
The proof from right to left is as follows.
\[
  c \refsto \eguard{r} \together d
 \Implies
  \eguard{r} \together c \refsto \eguard{r} \together \eguard{r} \together d
 \IFF
  \eguard{r} \together c \refsto \eguard{r} \together d
\]
\end{proof}

The literature on rely-guarantee often makes use of a refinement operator
$c \refsto_r d$
requiring $d$ to refine $c$ provided all environment steps satisfy $r$ \cite{CoJo07}.
More formally, the set of traces of $d$ restricted to those with all environment steps
satisfying $r$ is contained in the traces of $c$.
The set of traces of $d$ restricted to those with all environment steps satisfying $r$
corresponds to $\eguard{r} \together d$,
and hence the refinement $c \refsto_r d$ is equivalent to
\begin{eqnarray}\label{c-to-eguardr-d}
  c \refsto \eguard{r} \together d~.
\end{eqnarray}
Hence by making use of a richer set of commands,
one can avoid the need to introduce the family of refinement relations ``$\refsto_r$'' indexed by $r$.
Property (\ref{c-to-eguardr-d}) is subtly different to the refinement 
\begin{eqnarray}\label{envr-c-to-d}
  \envc{r} \together c \refsto d
\end{eqnarray}
but only for commands $c$ that restrict environment steps.
In general, (\ref{envr-c-to-d}) implies (\ref{c-to-eguardr-d}) but not vice versa.
To see that the reverse does not hold take $c$ to be $\eguard{r}$ and $d$ to be $\Chaos$.
Refinement (\ref{c-to-eguardr-d}) reduces to $\eguard{r} \refsto \eguard{r} \together \Chaos = \eguard{r}$,
which holds,
but (\ref{envr-c-to-d}) reduces to $\envc{r} \together \eguard{r} \refsto \Chaos$,
which does not hold in general because the left side is equivalent to $\eguard{r}$
by Lemma~\ref{eguard-env}.

\begin{theorem}
For a binary relation $r$, and commands $c$ and $d$,~~
\(
  \envc{r} \together c \refsto d ~~\implies~~ c \refsto \eguard{r} \together d~.
\)
\end{theorem}

\begin{proof}
\[
  \envc{r} \together c \refsto d
 \Implies*[as \strictconjunction\ is monotone]
  \eguard{r} \together \envc{r} \together c \refsto \eguard{r} \together d
 \IFF*[by Lemma~\ref{eguard-env}]
  \eguard{r} \together c \refsto \eguard{r} \together d
 \IFF*[by Lemma~\ref{eguardr-refsto}]
  c \refsto \eguard{r} \together d
\]
\end{proof}

This law illustrates the connection between environment restriction and
relies with respect to refinement, much like the connection between guards
and assertions for preconditions from the sequential refinement calculus.  
Given (normal) programs $c$ and $d$ that do not restrict the environment, 
making a rely assumption in the specification $c$ is equivalent to restricting
the environment in $d$ when showing refinement.

\subsection{Local variable blocks}\label{s:local-variables}

In the context of interference from concurrent programs,
a local variable cannot be affected by any programs
external to the block in which it is declared and 
the body of the local variable block cannot reference%
\footnote{We assume there is no aliasing of variable names 
otherwise restrictions on aliasing need to be included.}
(and hence modify) 
any more global variable with the same name.
The definition of a local variable block can be built from the language primitives defined earlier,
including the operator $c \hide x$, 
which has traces that are the same as $c$ except that $x$ is unrestricted
\refeqn{hide}.

The command $(\Local{x}{c})$ localises $x$ within $c$ and hence 
it restricts all environment steps of $c$ to not modify $x$,
i.e.\ they satisfy the relation $\id(x)$.
The program steps of the local variable block do not change the value of a
non-local occurrence of the variable $x$, i.e.\ $\pguard{\id(x)}$.
For a local variable block $(\variable{x}{c})$, 
$c$ is run in an environment in which $x$ is local.
To allow for allocation and deallocation of $x$,
the block may idle before and after execution (implicitly such idle steps
may modify the local $x$).
\begin{eqnarray*}
  \Local{x}{c} & \sdefs & \pguard{\id(x)} \together (c \together \eguard{\id(x)}) \hide x \\
  \variable{x}{c} & \sdefs & \Idle \SSeq (\Local{x}{c}) \SSeq \Idle 
\end{eqnarray*}
From this definition it follows that $(\variable{x}{c}) = \pguard{\id(x)} \together (\variable{x}{c})$ and
any refinement of $c$ within the block can assume there is no interference on $x$.
If the body of a local variable block guarantees the relation $g$ on every program step,
the block itself respects both $\id(x)$ and $g \hide x$,
where $(\sigma, \sigma') \in (g \hide x) \iff 
(\exists \sigma_1, \sigma_1' \spot x \ndres \sigma = x \ndres \sigma_1 \land x \ndres \sigma' = x \ndres \sigma_1' \land (\sigma_1,\sigma_1') \in g)$.
\begin{eqnarray*}
  (\variable{x}{\pguard{g} \together c}) & = & \pguard{g \hide x \int \id(x)} \together (\variable{x}{\pguard{g} \together c})
\end{eqnarray*}
An assumption about environment steps satisfying $r$ can be propagated into a local variable block
by ignoring the global variable $x$ (if any) and assuming the local $x$ cannot be changed by the environment.
\begin{eqnarray*}
  \envc{r} \together (\variable{x}{c}) & \refsto & (\variable{x}{\envc{(r \hide x) \int \id(x)} \together c})
\end{eqnarray*}

\newcommand{\pvariable}[2]{(\variable{#1}{#2})}
\newcommand{\pLocal}[2]{(\Local{#1}{#2})}

Consider the variable declaration command $d \sdef \pvariable{x}{c}$ in
an environment $e$, i.e.\ $d \pl e$.  
The intention is that $c$ may modify $x$, but that $e$
may not modify $d$'s $x$.  The former intention is implicitly
allowed by $d$.  The latter constraint is enforced by the environment guard
$\id(x)$ in $\pLocal{x}{c}$.  Typical approaches to local variables require
a syntactic constraint on $e$ that it never modify $x$, but this is
problematic if $e$ declares its own copy of $x$, for instance, if $d$ and
$e$ are instances of the same command.  In that case a further implicit or
explicit renaming scheme is introduced to remove the ambiguity.  Such
schemes create significant overhead if enforced fully (for instance,
environments in Plotkin-style semantics).  The approach taken here avoids
such syntactic concerns by ``hiding'' $c$'s modifications of $x$ from $e$.
This is given by $c \hide x$ in $\pLocal{x}{c}$, which removes the
modifications of $x$ in the program steps of $c$.  These modifications are
then replaced by a guarantee that $d$ does not modify $x$, i.e.\ $\pguard{\id(x)}$, as expected: if
$x$ \emph{is} (re)declared in $e$, $d$ does not modify it, since all
references to $x$ in $d$ are local.

More concretely, the traces of $\pvariable{x}{c_1} \pl \pvariable{x}{c_2}$ do not
interfere with each other on $x$.  The declaration of $x$ as local to each
program is a semantic restriction, rather than a syntactic restriction.

\subsection{Atomic abstract specifications and linearisability}\label{s:linearisability}

\newcommand{\PL}[2]{{\underset{#1}{\pl} #2}}

\renewcommand{\arelidle}[1]{\atomicrel{#1}}

We have previously discussed end-to-end specifications, which
allow implementations that comprise multiple steps that modify (potentially shared)
data.  In concurrent contexts in which many programs modify some central shared
data, end-to-end specifications are not appropriate: either a multi-stage
update to shared data may interfere with the other programs, or other
programs may modify the shared data after completing an operation,
falsifying the end-to-end requirement.  For situations where multiple
programs are expected to concurrently make calls that access/update shared data, 
some form of locking or other mechanism is required to ensure the data remains
consistent.  One of the most widely known correctness criterion for shared
data with atomic actions is linearisability \cite{linearisability}.
We demonstrate below how the wide-spectrum language is applicable in this
widely-used context.

To describe atomic execution it useful to define an abstract \emph{atomic
step} command $\arelidler$, which performs exactly one atomic step that
satisfies $r$, in addition to which it may idle before and after.
\begin{eqnarray}
  \arelidler & \sdefs & \Idle \SSeq \cpstep{r} \SSeq \Idle  \label{atomicrelidle}
\end{eqnarray}

\subsubsection{Concurrent objects}\label{s:objects}

The concept of a concurrent object is shared data $x$ which may be modified
by programs (indexed over the finite set $P$) performing operations (drawn from
set $OP$) concurrently.  This can be abstractly stated in the wide-spectrum
language using local variables and generalising binary parallel
composition.
\begin{equation}
\label{concobj}
	\variable{x}{\PL{p \in P}{\left(\Choice{}{op \in OP}{op} \right)\InfIter}}
\end{equation}
Each program repeatedly executes operations on $x$.
The interleaving of multiple modifications to $s$ results in undesirable
interference.  Traditionally such interference is avoided through the
introduction of locks, although this approach has drawbacks associated with
deadlock and unreliable processors.  Another approach is to use
atomic hardware primitives such as compare-and-swap (CAS) and retry loops
to avoid deadlock.  In either approach, the intuition is that the
operations appear to take effect atomically from the perspective of 
other programs.  This notion of correctness is formalised by Herlihy \&
Wing's definition of \emph{linearisability}.  Although the original
definition is in terms of histories of operation invocation and response
events (rather than traces of states), it is often equivalently recast in terms
of atomic specifications on abstract data, which are data refined to
(usually pointer-based) implementations in code.

\subsubsection{Atomic operations on concurrent objects}

An operation to perform an atomic step satisfying the relation $q$ on a data structure $x$ has 
the following form, which includes an assumption $p$ of the initial state and
a rely $r$ of the environment.
\begin{eqnarray}
	\assertp \SSeq~ \Frame{x}{\arelidle{q}} \together \envc{r}
	\label{lin-op-spec}
\end{eqnarray}
Note that finite stuttering is allowed in addition to the single $q$-step.
Any guarantee condition can be included as a conjunct within $q$.
As an example, below we define
$push$ and $pop$ operations on an abstract stack 
represented as a sequence $s$.
Recall that sequence concatenation is represented by the operator ``$\cat$'' and
$\seqT{v}$ represents the singleton sequence containing $v$.
\begin{eqnarray}
	push(s, v) 
	&\sdef&
	\Frame{s}{\arelidle{s' = \seqT{v} \cat s}} \together \envc{v' = v}
	\label{lin-push}
	\\
	pop(s, v) 
	&\sdef&
	\Frame{s,v}{\arelidle{
		s = \seqT{v'} \cat s' 
		\lor
		(s = s' = \eseq \land v' = \popempty)
	}} \together
	\envc{v' = v}
\end{eqnarray}
A call $push(s,v)$ atomically adds value $v$ to the top of the
stack.  An implementation of $push(s,v)$ may make any
finite number of stuttering steps before and after the main atomic step that
modifies $s$; this is important for implementations that may use a
different representation.  The $push$ operation
need not assume anything about the initial state
and assumes that the environment does not modify the parameter $v$
(and that the type of $s$ remains a stack, which we leave implicit). 
A call $pop(s, v)$ returns the popped
value in $v$, or the special value $\popempty$ if the stack is 
empty at the time the atomic step of the pop occurs. 

Encapsulating the stack as a concurrent object gives the following instance of
\refeqn{concobj}.
\[
	Stack \sdef 
	\variable{s}{\PL{p \in P}{ 
		\pvariable{v}{( push(s, v) \choice pop(s, v) )\InfIter}}}
\]
In this specification multiple programs concurrently modify $s$ by
atomically pushing to or popping from it, as recorded in the local variable $v$ of each process.

Compare the specification of atomic push on a stack in \refeqn{lin-push}
to that using an end-to-end specification in \refeqn{seq-push}.
In the former no assumption is made about $s$ (except that implicitly it
respects its type), while in the latter $s$ is assumed not to be changed by the environment.  The
former is a weaker assumption about the environment, and hence more
desirable in general, but with the trade-off that the push itself must be
completed in exactly one program step.  The specification in
\refeqn{seq-push} allows multiple steps to achieve the specification since
it may assume no interference on $s$.
A further point of difference is that a guarantee is redundant for the atomic specification
\refeqn{lin-push} since it takes only one step.

\section{Progress properties}\label{s:progress}

In addition to ``safety'' properties such as establishing an end-to-end relation,
concurrent programs are often required to establish ``progress''
properties regarding overall behaviour.
In Section~\ref{s:temporal-logic} temporal logic is recast into our refinement framework,
and then used in Section~\ref{s:nonblocking}
to express the well-known ``lock-free'' property,
and two related concepts.

\subsection{Temporal logic}\label{s:temporal-logic}

Linear temporal logic (LTL) \cite{Pnueli77} describes properties of infinite traces.  
It is straightforward
to encode a temporal logic formula $f$ as a command $\encode{f}$ 
such that
a trace $tr$ satisfies $f$ if and only if $tr \in \encode{f}$
and hence
a command $c$ satisfies $f$ if and only if $\encode{f} \refsto c$.
As well as standard LTL formulae,
one can include formulae representing a single program or environment step
satisfying a relation $r$, 
for which the notations $\TLPstepr$ and $\TLEstepr$, respectively, are used.
This makes it possible to devise temporal logic formulae that relate program and environment behaviour.
In the following, 
$p$ is a single-state predicate, 
$r$ is a binary relation on states,
$h$ is a temporal logic formula,
and
$x$ is a bound variable from a recursion.
\begin{displaymath}
 \begin{array}[t]{cc}
  \mbox{LTL formula $f$} & \encode{f}
  \\
  p & \cgdp \Seq \Abort 
  \\
  \TLPstepr & \cpstepr \Seq \Abort 
  \\
  \TLEstepr & \cestepr \Seq \Abort 
  \\
  h_1 \land h_2 & \encode{h_1} \sqcup \encode{h_2}
  \\
  h_1 \lor h_2 & \encode{h_1} \sqcap \encode{h_2}
 \end{array}
 \hspace*{4em}
 \begin{array}[t]{cc}
  \mbox{LTL formula $f$} & \encode{f} 
  \\
  \cnext h & (\cpstepd \nondet \cestepbotd) \Seq \encode{h}
  \\
  \eventually h & \Fin{(\cpstepd \nondet \cestepbotd)} \Seq \encode{h}
  \\
  \always h & \mu x \spot \encode{h} \sqcup \cnext x
  \\
  x & x
 \end{array}
\end{displaymath}

The encoding of the primitive single state LTL formula $p$ 
requires $p$ to hold in the initial state,
after which any behaviour ($\Abort$) is possible.  
In addition to this familiar LTL formula, in our language we 
allow relations $r$ to be LTL formulae, stating whether the program or
environment must establish $r$.  This distinction is used below to require
that the environment always take some step.
Conjunction and disjunction of formulae straightforwardly use intersection ($\sqcup$)
and union ($\sqcap$) over sets of traces.  The LTL formula for next-step $\cnext
h$ requires exactly one step (program or environment) to occur and then
$h$ to be satisfied.  The LTL formula for eventually-holds $\eventually h$
requires $h$ to hold after a finite number of steps.  The encoding for
always-holds ($\always h$) uses a typical recursive encoding where $h$ is
required to hold in every step.

These definitions highlight the expressive power of a language based on
primitives, and are exploited below to describe some common progress
properties of concurrent programs.
Further ways to develop a temporal logic for our wide spectrum language,
including
interval temporal logic, are outside the scope of this paper.

\subsection{Non-blocking operations}\label{s:nonblocking}

\newcommand{\ofop}{op_{of}}
\newcommand{\lfop}{op_{lf}}
\newcommand{\wfop}{op_{wf}}

As it is usually conceived, \emph{lock-freedom} allows potential non-termination of an
operation provided the system as a whole continues to make
progress.  That is, the non-terminated operation does not ``lock'' the system, and
an infinite number of operations (executed by other parallel processes) may still complete.
For a data structure $s$, for which the shared data is modified at most
once per operation (successful or otherwise), system-wide progress can
be detected by the data structure changing; more specifically
in our context, by a step $\TLEstep{s' \neq s}$.

Related concepts include the stronger \emph{wait-freedom}, that
means every program completes regardless of interference, and the weaker
\emph{obstruction-freedom}, that states that programs need make progress only
in the absence of other programs \cite{obstruction-free}.  
Given an atomic operation $op$ that operates on
data structure $s$, 
we can specify these
properties straightforwardly using distinguished environment steps and
LTL.
\begin{eqnarray}
	\ofop
	&\sdef&
	op
	\choice
	(\pguard{\id} \together
		\calways \ceventually 
			\TLEstepd 
	)
	\\
	\lfop
	&\sdef&
	op
	\choice
	(\pguard{\id} \together
		\calways \ceventually \TLEstep{s' \neq s}
	)
	\\
	\wfop
	&\sdef&
	op
\end{eqnarray}
An obstruction-free 
operation is weaker than a lock-free operation: in addition to
successful completion of the operation, an implementation of $\ofop$
may 
fail to terminate if the environment continues to take steps 
(always eventually the trace contains an environment step).  
The restriction $\pguard{\id}$ ensures that the implementation may not change the state
in this case. 
The lock-free specification $\lfop$ is similar, but an implementation
may fail to terminate only if
the environment continually modifies $s$, indicating system-wide
progress.
A valid implementation of $\lfop$ therefore includes
a potentially infinite loop that repeats whenever interference is
detected, and otherwise completes the operation.
Wait-freedom requires termination regardless of the environment; a 
wait-free operation is the default case.
This classification makes it clear that $\ofop \refsto \lfop \refsto
\wfop$, as expected.
If the environment performs no steps, $\ofop$ implements $op$
and 
if the environment does not change $s$, $\lfop$ implements $op$.
\begin{eqnarray*}
  op & \refsto & \ofop \together \eguard{\emptyrel} \\
  op & \refsto & \lfop \together \eguard{\id(s)}
\end{eqnarray*}

This classification is simpler than many in the literature 
(though less general than Dongol~\cite{DongolFormalisingNonBlocking}).
We argue that this is achieved by using Aczel traces to specify the
environment, and applying LTL where appropriate.  Note that
properties
of guarantees and LTL developed in different contexts may be reused to
tackle nonblocking programs.

\section{Conclusions}
\label{s:conclusions}

Our goal in designing the semantic model presented in this paper is to provide 
an approach that supports a concurrent wide-spectrum language
that includes both programming language constructs and specification constructs,
including constructs that support rely and guarantee conditions.
A novelty of our approach is that we make use of a set of primitive commands,
including $\cpstepr$, $\cestepr$ and $\cgdp$,
that are quite close to the atomic steps of the underlying semantics
and provide an expressive framework for defining diverse concepts from concurrency theory.  
More complex commands are defined using these primitive commands and the primitive operators.
The use of such basic primitives allows one to make fine distinctions in the semantics,
e.g. we considered three ``no-op'' commands:
$\Nil$, 
$\Skip$, 
and
$\Idle$.

The semantics of the primitives is given in terms of Aczel traces \cite{Aczel83},
similar to the approach of De Roever \cite{DeRoever01}
but extended to support aborting and infeasible specification constructs.
Alternative approaches to specifying the core language are possible,
for example using an operational semantics 
that distinguishes program and environment steps.
As long as the properties of the primitives are maintained,
the definition of the wide-spectrum language constructs in terms of the primitives
requires no change.

Brookes gives a semantics for a parallel language in terms of \emph{action traces}~\cite{BrookesSepLogic07}.
He gives a set of ten primitive \emph{actions} that form the atomic units of a program's execution.
Because he supports concurrent separation logic, 
his model of the state includes a store, a heap and a finite set of resources.
Given such a model of the state, 
his primitive actions can all be defined in terms of a command of the form $\cpstepr$,
in which the relation $r$ determines the effect of the primitive.
More generally, the primitives could be defined as a command of the form $\optr{r}$,
which allows the step to be elided for initial states in which the command has no effect on the state.

In \cite{HayesJonesColvin14TR,FACJexSEFM-14} a
refinement calculus for developing programs in the rely-guarantee style
is presented, using an underlying semantics similar to (but less general
than) that presented here.  Many of the laws developed in that paper for
dealing with assumptions about initial states and the environment may be
carried over to reasoning about those aspects of atomic
specifications, which is an advantage of a unified
concurrency framework.

\subparagraph*{Acknowledgements.}

This research was supported by
ARC grant
DP130102901.
This paper has benefited from feedback from
Julian Fell,
Cliff Jones,
Graeme Smith,
Kim Solin
and
Kirsten Winter.

\bibliographystyle{plain}
\bibliography{parallel}

\end{document}